\documentclass[aps,pra,showpacs,showkeys,reprint,floatfix]{revtex4-1}
\usepackage{graphicx}
\usepackage{amsfonts}
\usepackage{amsmath}
\usepackage{amssymb}
\usepackage{amsthm}
\usepackage{mathrsfs}

\newcommand{\desangle}[2]{\theta^\downarrow_{#2}(#1)}
\newcommand{\desabsangle}[2]{|\theta|^\downarrow_{#2}(#1)}
\newcommand{\unorderangle}[2]{\theta_{#2}^{#1}}
\newcommand{\desbarebra}[2]{\langle\phi^\downarrow_{#2}(#1)}
\newcommand{\desbra}[2]{\desbarebra{#1}{#2}|}
\newcommand{\desket}[2]{|\phi^\downarrow_{#2}(#1)\rangle}
\newcommand{\unorderbarebra}[2]{\langle\phi_{#2}^{#1}}
\newcommand{\unorderbra}[2]{\unorderbarebra{#1}{#2}|}
\newcommand{\unorderket}[2]{|\phi_{#2}^{#1}\rangle}
\newcommand{\vecmu}[1]{\vec{\mu}^{[#1]}}
\newcommand{\enorm}[1]{\nu (#1)_{\vec{\mu}}}
\newcommand{\spcenorm}[2]{\nu (#1)_{\vecmu{#2}}}
\newcommand{\emetric}[2]{d_{\vec{\mu}}(#1,#2)}

\newcommand{\Nenorm}[1]{\nu (#1)^\bigtriangledown_{\vec{\mu}}}
\newcommand{\Nspcenorm}[2]{\nu (#1)^\bigtriangledown_{\vecmu{#2}}}
\newcommand{\Nemetric}[2]{d^\bigtriangledown_{\vec{\mu}}(#1,#2)}

\newcommand{\ecomm}[2]{{\mathfrak C}_{\vec{\mu}}(#1,#2)}
\newcommand{\deseig}[2]{\lambda^\downarrow_{#2}(#1)}

\newcommand{\deseigket}[2]{|\xi^\downarrow_{#2}(#1)\rangle}
\newcommand{\desabseig}[2]{s^\downarrow_{#2}(#1)}

\newtheorem{definition}{Definition}
\newtheorem{lemma}{Lemma}
\newtheorem{theorem}{Theorem}
\newtheorem{corollary}{Corollary}
\newtheorem{remark}{Remark}

\begin{document}

\title{Metrics On Unitary Matrices And Their Application To Quantifying The
 Degree Of Non-Commutativity Between Unitary Matrices}

\author{H.~F. Chau}\email{hfchau@hkusua.hku.hk}
\affiliation{Department of Physics and Center of Computational and Theoretical
 Physics, University of Hong Kong, Pokfulam Road, Hong Kong}
\date{\today}

\begin{abstract}
 By studying the minimum resources required to perform a unitary
 transformation, families of metrics and pseudo-metrics on unitary matrices
 that are closely related to a recently reported quantum speed limit by the
 author are found.  Interestingly, this family of metrics can be naturally
 converted into useful indicators of the degree of non-commutativity between
 two unitary matrices.
\end{abstract}

\pacs{03.65.Aa, 02.10.Yn, 03.65.Ta, 03.67.-a}
\keywords{Eigenvalue Perturbation Theory, Measurement Of Non-commutativity,
 Metrics For Unitary Matrices, Quantum Information Processing}

\maketitle

\section{Introduction \label{Sec:Intro}}
 Quantum information processing is the study of methods and efficiency in
 storage, manipulation and conversion of information represented by quantum
 states.  Many quantum information theoretic concepts are closely related to
 geometry.  For instance, trace distance and fidelity, which come out of the
 study of distinguishability between quantum states, are closely linked with
 Bures and Fubini-Study metrics.  (See, for example,
 Ref.~\cite{geometry_of_states_book} for an overview.)  A few quantum codes can
 be constructed by algebraic-geometric means~\cite{AG_code}.  And finding the
 optimal quantum circuit can be regarded as the problem of finding the shortest
 path between two points in a certain curved
 geometry~\cite{geometric_computing}.
 
 Recently, a few metrics on unitary operators with quantum information
 applications were found.  For example, Johnston and Kribs introduced the $k$th
 operator norm of an operator acting on a bipartite system by considering the
 action of the operator on bipartite states with Schmidt rank less than or
 equal to $k$.  The $k$th operator norm can be used to study bound entanglement
 of Werner states as well as to construct several new entanglements
 witnesses~\cite{qip_operator_norm,further_qip_operator_norm}.  Rastegin
 studied the partitioned trace distance which shares similar properties with
 the standard trace distance~\cite{singular_value_trace_distance}.  The
 partitioned trace distance shines new light on exponential
 indistinguishability and hence can be used to investigate certain quantum
 cryptographic problems~\cite{more_singular_value_trace_distance}.
 Interestingly, both the $k$th operator norm and the partitioned trace distance
 are related to the Ky Fan
 norm~\cite{qip_operator_norm,singular_value_trace_distance}.

 By asking the question about the minimum resources needed to perform a unitary
 transformation (a question of quantum information processing favor), I report
 families of related metrics and pseudo-metrics on the set of unitary matrices
 (a result of geometric nature) in this paper.  In this regard, these matrices
 and pseudo-metrics are very different from trace distance, partitioned trace
 distance and $k$th operator norm as the latter are more closely related to a
 different quantum information processing problem, namely, the
 distinguishability of states and operators.

 An interesting consequence of the discovery of this new family of metrics on
 the set of unitary matrices is that it gives refined measures of the degree of
 non-commutativity between two unitary matrices beyond the standard yes-or-no
 answer.  Remarkably, while quantifying the level of closeness between two
 quantum states can be done by tools such as trace distance and fidelity (see,
 for example, Refs.~\cite{geometry_of_states_book,Nielsen_Chuang} for an
 introduction on this matter), little has been done on the quantification of
 the degree of non-commutativity between two unitary operators.

 Actually, the metrics and pseudo-metrics reported in this paper are
 constructed from certain linear combinations of the absolute values of the
 arguments of eigenvalue of a unitary matrix.  To prove that these
 constructions are indeed metrics and pseudo-metrics, specific inequalities
 concerning the absolute values of the arguments of eigenvalue of the unitary
 matrices $U, V$ and $U V$ have to be established.  The precise statements to
 be proven can be found in Definitions~\ref{Def:form} and~\ref{Def:norm_metric}
 as well as Theorem~\ref{Thrm:angle_bound} below.  Note that various authors
 had shown the validity of a similar inequality, with the weighted sum of the
 absolute values of the arguments of eigenvalue for each unitary matrix being
 replaced by the largest argument of the eigenvalue for that matrix, provided
 that an additional condition constraining the arguments of eigenvalues of $U$
 and $V$ is satisfied~\cite{convex_hull_bound,direct_extension,q_calculus,%
principal_bound}.
 However, it is not clear how to modify their proofs to show the validity of
 the inequalities stated in Theorem~\ref{Thrm:angle_bound}.  In fact, only a
 few results on the relationship between arguments of eigenvalues of unitary
 matrices $U, V$ and $U V$ are known.  Most mathematical works along this
 general direction concentrate on the study of spectral variations of normal
 matrices on the complex plane (rather than arguments of the eigenvalues) as
 well as relations between the eigenvalues of Hermitian matrices $H_1, H_2$ and
 $H_1 + H_2$ (known as the Weyl's problem).  (See, for example,
 Refs.~\cite{matrix_analysis_book,eigenvalue_book,Weyl_review} for
 comprehensive surveys.)  Besides, techniques used to tackle the Weyl's
 problem, such as the min-max principle stated in Sec.~III.1 of
 Ref.~\cite{matrix_analysis_book}, cannot be easily adapted to the case of
 unitary matrices.  In this regard, the proof of Theorem~\ref{Thrm:angle_bound}
 is also of mathematical interest.

 I begin by asking what is the minimum resources needed to perform a unitary
 transformation in Sec.~\ref{Sec:Background}.  Motivated by the result of this
 quantum information theoretic question, I define two closely related families
 of binary operations on the set of unitary matrices $\emetric{\cdot}{\cdot}$
 and $\Nemetric{\cdot}{\cdot}$, where $\vec{\mu}$ is a real-valued vector
 satisfying a technical condition, in Sec.~\ref{Sec:Def}.  Then, in
 Sec.~\ref{Sec:Properties}, I prove that these two families of binary
 operations are families of metrics and pseudo-metrics on the set of all
 unitary matrices, respectively.  In Sec.~\ref{Sec:Non-Commutativity}, I apply
 the metric introduced in Sec.~\ref{Sec:Def} to measure the degree of
 non-commutativity between two unitary matrices.  Finally, I give a summary and
 outlook in Sec.~\ref{Sec:Discuss}.

\section{Minimum Resources Needed To Perform A Unitary Transformation
 \label{Sec:Background}}
 Two commonly studied problems in the field of quantum information processing
 are the maximum efficiency of a particular quantum information processing
 operation as well as the minimum resources needed to carry out such operation
 allowed by the known laws of nature~\cite{phy_limit}.  Since a unitary
 operator $U$, whose eigenvalues can be written in the form $e^{i\theta_j}$'s
 with $\theta_j \in (-\pi,\pi]$, is the result of time evolution of an
 Hamiltonian $H$, one may naively use the values of $|\theta_j|$'s as
 indicators of the minimum resources (in terms of the product of the average
 energy of the system and the evolution time needed) required to implement $U$.
 Nonetheless, this idea has to be polished as there is no physical meaning for
 the reference energy level and the overall phase of a unitary operator has no
 effect when applied to a density matrix.

 To refine the above idea, I use the following result.

\begin{definition}
 Let $H$ be an $n$-dimensional Hermitian matrix.  I follow the notation in
 Ref.~\cite{matrix_analysis_book} by denoting the eigenvalues and singular
 values of $H$ arranged in descending order by $\deseig{H}{j}$'s and
 $\desabseig{H}{j}$'s respectively, where the index $j$ runs from $1$ to $n$.
 I denote the normalized eigenvector of $H$ with eigenvalue $\deseig{H}{j}$ by
 $\deseigket{H}{j}$.
 \label{Def:Hermitian_form}
\end{definition}

 Recently, I showed that given a time-independent Hamiltonian $H$ and a
 normalized initial state $|\phi\rangle = \sum_j \alpha_j \deseigket{H}{j}$,
 the time $\tau$ needed to evolve $|\phi\rangle$ to a state in its orthogonal
 subspace satisfies~\cite{preprint}
\begin{equation}
 \tau \geq \frac{\hbar}{A \sum_j |\alpha_j|^2 |\deseig{H}{j} - x|}
 \label{E:absolute1}
\end{equation}
 for any $x\in {\mathbb R}$.   Here $A$ is a universal positive constant
 independent of $H$ and $|\phi\rangle$.  Note that $\sum_j |\alpha_j|^2
 |\deseig{H}{j} - x|$ is minimized when $x$ equals the median energy of the
 system $M$.  Recall that the median energy $M$ satisfies
\begin{subequations}
\begin{equation}
 \sum_{j\colon \deseig{H}{j} \geq M} |\alpha_j|^2 \geq \frac{1}{2}
 \label{E:Median_def1}
\end{equation}
 and
\begin{equation}
 \sum_{j\colon \deseig{H}{j} \leq M} |\alpha_j|^2 \geq \frac{1}{2} .
 \label{E:Median_def2}
\end{equation}
\end{subequations}
 So, $M$ need not be unique; and any $M$ obeying the above two equations can
 minimize $\sum_j |\alpha_j|^2 |\deseig{H}{j} - x|$.  From
 Eq.~(\ref{E:absolute1}), I get
\begin{equation}
 \tau \geq \frac{\hbar}{A \sum_j |\alpha_j|^2 |\deseig{H}{j} - M|} \equiv
 \frac{\hbar}{A {\mathscr D}E(H,|\phi\rangle)}
 \label{E:absolute2}
\end{equation}
 where ${\mathscr D}E(H,|\phi\rangle)$ is the so-called average absolute
 deviation from the median of the energy of the system.  Since
 Eq.~(\ref{E:absolute2}) turns out to be the best possible lower bound on
 $\tau$ given only ${\mathscr D}E(H,|\phi\rangle)$, I interpreted
 ${\mathscr D}E(H,|\phi\rangle)$ as an indicator of the maximum quantum
 information processing rate of a system~\cite{preprint}.

 The above result can be used to prove the theorem below which quantifies the
 resources needed to perform a unitary transformation.  The detailed proof can
 be found in Appendix~\ref{Sec:proof_of_interpretation}.

\begin{theorem}
 Let $( \alpha_j )_{j=1}^n$ be a sequence of complex numbers obeying
 $\sum_{j=1}^n |\alpha_j|^2 = 1$ and $|\alpha_1|^2 \geq |\alpha_2|^2 \geq
 \cdots \geq |\alpha_n|^2$.  And let $U$ be a given $n\times n$ unitary matrix.
 Then
 \begin{widetext}
 \begin{equation}
  R = \min_{x\in [0,2\pi)} \quad \min_{H t \colon \exp ( - i H t / \hbar) =
  e^{i x} U} \quad \max_{|\phi\rangle \in C(H,(\alpha_j))} \frac{A
  {\mathscr D}E(H,|\phi\rangle) t}{\hbar}
  \label{E:min_max_interpretation}
 \end{equation}
 \end{widetext}
 exists, where $C(H,(\alpha_j))$ is the set of normalized state kets in the
 form $\sum_{j=1}^n \alpha_j \deseigket{H}{P(j)}$ for some permutation $P$ of
 $\{ 1,2,\cdots ,n \}$.  Moreover, the $H$ and $|\phi\rangle$ that attain the
 extremum in Eq.~(\ref{E:min_max_interpretation}) can be chosen to have $0$
 median system energy and $\deseig{H}{j} t / \hbar \in (-\pi,\pi]$ for all $j$.
 In particular, there exists a Hamiltonian $H$ such that
 \begin{equation}
  R = \frac{A t \sum_{j=1}^n |\alpha_j|^2 \desabseig{H}{j}}{\hbar} .
  \label{E:min_max_expression}
 \end{equation}
 \label{Thrm:interpretation}
\end{theorem}

 Theorem~\ref{Thrm:interpretation} can be understood physically as follows.
 Recall that ${\mathscr D}E(H,|\phi\rangle)$ is an indicator of the maximum
 quantum information processing rate.  Since $U = e^{-i H t / \hbar}$ up to an
 overall phase, one may implement $U$, for instance, by evolving the system
 with a slow quantum information processing rate $H$ for a long time or
 alternatively by evolving the system with a fast quantum information
 processing rate $H$ for a short time.  In this respect, it is the product of
 ${\mathscr D}E(H,|\phi\rangle)$ and evolution time $t$ that characterizes the
 resources needed to implement $U$.  More precisely, the value of $R$ defined
 by Eq.~(\ref{E:min_max_interpretation}) in Theorem~\ref{Thrm:interpretation}
 can be regarded as an indicator of the least amount of resources needed to
 carry out the unitary transformation $U$ on the worst possible $|\phi\rangle$
 in the form $\sum \alpha_j \deseigket{H}{P(j)}$ for some permutation $P$ of
 $\{1,2,\ldots ,n\}$.  The larger the value of $R$, the more the average
 absolute deviation from the median of the energy of the system times the
 evolution time is needed to carry out the transformation on the worst possible
 $|\phi\rangle$.

\section{Defining Families Of Quantum Information Theory Inspired Metrics And
 Pseudo-metrics On The Set Of Unitary Matrices \label{Sec:Def}}
 Based on the quantum information theoretic analysis in
 Sec.~\ref{Sec:Background} and the fact that $R$ in
 Eqs.~(\ref{E:min_max_interpretation}) and~(\ref{E:min_max_expression}) are
 essentially a weighted sum of singular values of the operator $i\log U - x I$
 for some $x\in {\mathbb R}$, I propose a number of measures of the minimum
 amount of the average absolute deviation from the median of the energy of the
 system times the evolution time required to implement a unitary operation.  I
 begin by introducing a few notations.

\begin{definition}
 Let $U$ be an $n$-dimensional unitary matrix.  Generalizing the convention
 adopted in Ref.~\cite{matrix_analysis_book}, I denote the arguments (from now
 on, all arguments in this paper are in principal values unless otherwise
 stated) of the eigenvalues of $U$ arranged in descending order by
 $\desangle{U}{j}$'s (where the index $j$ runs from $1$ to $n$).  That is to
 say, $U = \sum_j e^{i \desangle{U}{j}} \desket{U}{j}\desbra{U}{j}$ where
 $\desangle{U}{j} \in (-\pi,\pi]$ and $\desket{U}{j}$ is a normalized
 eigenvector of $U$ with eigenvalue $e^{i \desangle{U}{j}}$.  Similarly, I
 denote by $\desabsangle{U}{j}$'s the absolute values of the arguments of the
 eigenvalue of $U$ arranged in descending order.  Occasionally, I need to refer
 to the arguments of the eigenvalues of $U$ without any specific order using
 the notation $\unorderangle{U}{j}$'s.  And in this case, I denote the
 corresponding normalized eigenvector by $\unorderket{U}{j}$.
 \label{Def:form}
\end{definition}

 For example, let $U = \left[ \begin{array}{cc} 1 & 0 \\ 0 & -i \end{array}
 \right]$.  Then, $\desangle{U}{1} = 0$, $\desangle{U}{2} = -\pi/2$,
 $\desabsangle{U}{1} = \pi/2$ and $\desabsangle{U}{2} = 0$.

 Note further that for any $n$-dimensional unitary matrix $U$,
\begin{subequations}
\begin{equation}
 \desangle{U^{-1}}{j} = - \desangle{U}{n-j+1} ,
 \label{E:inverse_U}
\end{equation}
\begin{equation}
 \desabsangle{U^{-1}}{j} = \desabsangle{U}{j} ,
 \label{E:inverse_U_abs}
\end{equation}
and
\begin{equation}
 \desangle{U}{j} \leq \desabsangle{U}{j} .
 \label{E:des_desabs}
\end{equation}
\end{subequations}

\begin{definition}
 Let $a \in {\mathbb R}$ and $U = \sum_j e^{i \desangle{U}{j}}
 \desket{U}{j}\desbra{U}{j}$ be a unitary matrix.  Then, $U^a$ is defined to be
 the unitary matrix $\sum_j e^{i a \desangle{U}{j}}
 \desket{U}{j}\desbra{U}{j}$.
 \label{Def:matrix_power}
\end{definition}

\begin{definition}
 Let $U,V$ be two $n$-dimensional unitary matrices.  Let $\vec{\mu} = (\mu_1,
 \mu_2,\ldots ,\mu_n) \neq \vec{0}$ with $\mu_1 \geq \mu_2 \geq \cdots \geq
 \mu_n \geq 0$.  I define
 \begin{equation}
  \enorm{U} = \sum_{j=1}^n \mu_j \desabsangle{U}{j} ,
  \label{E:norm_def}
 \end{equation}
 \begin{equation}
  \emetric{U}{V} = \enorm{U V^{-1}} ,
  \label{E:metric_def}
 \end{equation}
 \begin{equation}
  \Nenorm{U} = \min_{x\in {\mathbb R}} \enorm{e^{i x} U} ,
  \label{E:Nnorm_def}
 \end{equation}
 and
 \begin{equation}
  \Nemetric{U}{V} = \min_{x\in {\mathbb R}} \emetric{e^{i x} U}{V} .
  \label{E:Nmetric_def}
 \end{equation}
 \label{Def:norm_metric}
\end{definition}

 I am going to prove in Sec.~\ref{Sec:Properties} that $\emetric{\cdot}{\cdot}$
 and $\Nemetric{\cdot}{\cdot}$ are a metric and pseudo-metric, respectively.
 This justifies the use of the notations.  On the other hand, $\enorm{\cdot}$
 and $\Nenorm{\cdot}$ are not a norm or pseudo-norm.  This is because $U + V$
 may not be unitary so that $\enorm{U + V}$ and $\Nenorm{U +V}$ are
 ill-defined.  Hence, it does not make sense to talk about validity of the
 (additive) triangle inequality for $\enorm{\cdot}$ and $\Nenorm{\cdot}$.
 Nevertheless, I am going to prove in Sec.~\ref{Sec:Properties} that
 $\enorm{\cdot}$ and $\Nenorm{\cdot}$ satisfy the multiplicative triangle
 inequality.

 Note further that $\enorm{U}$ is a weighted sum of the singular values of a
 certain infinitesimal generator of the unitary operator $U$.  And from
 Theorem~\ref{Thrm:interpretation}, I know that $\Nenorm{U}$ is indeed a
 measure of the minimum average absolute deviation from the median energy of
 the system times the evolution time needed to perform the unitary
 transformation $U$ in the sense that
\begin{equation}
 \Nenorm{U} = \min_{x\in [0,2\pi)} \quad \min_{H t \colon \exp ( -i H t /
 \hbar) = e^{i x} U} \quad \max {\mathscr D} E(H,|\phi\rangle) \ t ,
 \label{E:Nenorm_interpretation}
\end{equation}
 where the maximum is taken over all state kets $|\phi\rangle$ in the form
 $\sum_{j=1}^n \alpha_j \deseigket{H}{P(j)}$ with $|\alpha_j|^2 = \mu_j /
 \sum_k \mu_k$ and $P$ is a permutation of $\{ 1,2,\ldots,n\}$.  In contrast,
 $\enorm{U}$ can be regarded as an ``unoptimized'' version of $\Nenorm{U}$ in
 the sense that
\begin{equation}
 \enorm{U} = \min_{H t \colon \exp ( -i H t / \hbar) = U} \quad \max
 {\mathscr D}E(H,|\phi\rangle) \ t ,
 \label{E:enorm_interpretation}
\end{equation}
 where the maximum is taken over the same set as in the maximum in
 Eq.~(\ref{E:Nenorm_interpretation}).  In fact, the global phase of $U$, which
 has no physical meaning, affects the value of $\enorm{U}$ but not
 $\Nenorm{U}$.  Nevertheless, readers will find out in
 Sec.~\ref{Sec:Non-Commutativity} that $\enorm{\cdot}$ is useful in studying
 the degree of non-commutativity between two unitary operators.

 Three important points are stated.  First, one should investigate the
 properties of Eqs.~(\ref{E:norm_def})--(\ref{E:Nmetric_def}) for a general
 $\vec{\mu}$ in order to get a more complete picture on the minimum resources
 needed to evolve different kind of initial states $|\phi\rangle$'s by the
 unitary operator $U$.

 Second, an inequality giving a lower evolution time bound like the one stated
 in Eq.~(\ref{E:absolute2}) is sometimes known as a quantum speed
 limit~\cite{phy_limit}.  In addition to Eq.~(\ref{E:absolute2}), the other two
 important quantum speed limits are the so-called time-energy uncertainty
 relation $\tau \geq \pi\hbar / 2 \Delta E$ where $\Delta E$ is the standard
 deviation of the energy of the system and the Margolus-Levitin bound $\tau
 \geq \pi\hbar / 2(E - E_0)$ where $E$ and $E_0$ are the average energy and
 ground state energy of the system, respectively~\cite{PHYSCOMP96,max_speed}.
 One may define similar measures of the minimum resources needed to implement a
 unitary operator $U$ based on these two quantum speed limits.  Unfortunately,
 it is not very fruitful to pursue in this direction because these measures are
 not metrics.  (Note that the absence of metric structures here only refers to
 the measure of minimum resources required to perform a unitary operation based
 on these quantum speed limits.  It does not mean that the two quantum speed
 limits have no geometrical meaning.  In fact, the time-energy uncertainty
 relation itself is closely related to the Bures metric~\cite{Uhlmann}.)

 Third, $\Nenorm{\cdot}$ is related to the distance between two quantum states
 as follows.  Suppose $|\Psi_1\rangle$ and $|\Psi_2\rangle$ are two $n$-leveled
 pure states.  Recall that the Bures angle between them is $\chi = \left|
 \langle\Psi_1 | \Psi_2\rangle \right|$.  Then, amongst all the unitary
 transformations $U$ obeying $U |\Psi_1\rangle = |\Psi_2\rangle$, the one that
 minimizes $\Nenorm{U}$ satisfies $\desangle{U}{1} = \chi$, $\desangle{U}{n} =
 -\chi$ and $\desangle{U}{j} = 0$ for $1 < j < n$.  (One quick way to see this
 is that the required $U$ only need to perform rotation on the Hilbert space
 spanned by $|\Psi_1\rangle$ and $|\Psi_2\rangle$.  Then the problem can be
 handled by standard minimizing techniques.)

\section{Properties Of $\enorm{\cdot}$, $\emetric{\cdot}{\cdot}$,
 $\Nenorm{\cdot}$ And $\Nemetric{\cdot}{\cdot}$ \label{Sec:Properties}}
 Since $\vec{\mu} \equiv (\mu_1,\mu_2,\cdots ,\mu_n) = \sum_{j=1}^{n-1} (\mu_j
 - \mu_{j+1}) \vecmu{j} + \mu_n \vecmu{n}$ where
\begin{equation}
 \vecmu{j} = ( \overbrace{1,1, \cdots, 1}^{j \textrm{~entries}},0,0,\cdots ,0)
 ,
 \label{E:lambda_basic_def}
\end{equation}
 I conclude that 
\begin{subequations}
\begin{equation}
 \enorm{U} = \sum_{j=1}^{n-1} (\mu_j - \mu_{j+1}) \spcenorm{U}{j} + \mu_n
 \spcenorm{U}{n}
 \label{E:enorm_decomposition}
\end{equation}
 and
\begin{equation}
 \Nenorm{U} \geq \sum_{j=1}^{n-1} (\mu_j - \mu_{j+1}) \Nspcenorm{U}{j} + \mu_n
 \Nspcenorm{U}{n} .
 \label{E:Nenorm_decomposition}
\end{equation}
\end{subequations}
 Furthermore, the coefficients in the R.H.S. of the above two equations are
 non-negative.  It turns out that one can deduce many properties of
 $\enorm{\cdot}$ from the properties of $\spcenorm{\cdot}{j}$'s.

 Here I list a few basic properties of $\enorm{\cdot}$, $\Nenorm{\cdot}$,
 $\emetric{\cdot}{\cdot}$ and $\Nemetric{\cdot}{\cdot}$.  The proofs are left
 to the readers as they are straightforward.  (Note that all $U$'s and $V$'s
 appear below are $n$-dimensional unitary matrices.)
\begin{itemize}
 \item $\enorm{\cdot}$ and $\Nenorm{\cdot}$ are both functions from $U(n)$, the
  set of all $n\times n$ unitary matrices, to $[0,\pi \sum_{j=1}^n \mu_j]$.  In
  fact, $\enorm{\cdot}$ is a surjection with $\enorm{U} = \pi \sum_j \mu_j$ for
  all $\vec{\mu}$ if and only if $U = -I$.
 \item $\enorm{U} = 0$ and $\Nenorm{U} = 0$ if and only if $U = I$ and $U =
  e^{i x} I$ for some $x\in {\mathbb R}$, respectively.
 \item $\Nenorm{e^{i x} U} = \Nenorm{U}$ for all $x\in {\mathbb R}$.
 \item $\enorm{U^{-1}} = \enorm{U} = \enorm{V U V^{-1}}$ and $\Nenorm{U^{-1}} =
  \Nenorm{U} = \Nenorm{V U V^{-1}}$.
 \item $\nu (U)_{a\vec{\mu}}  = a \enorm{U}$ and $\nu (U)^\bigtriangledown_{a
  \vec{\mu}} = a \Nenorm{U}$ for all $a > 0$.
 \item $\enorm{U^b} \leq |b| \enorm{U}$ and $\Nenorm{U^b} \leq |b| \Nenorm{U}$
  for all $b\in {\mathbb R}$.  Moreover, the equalities hold if $|b|
  \desabsangle{U}{1} \leq \pi$.
 \item $\Nspcenorm{U}{2} = 2\Nspcenorm{U}{1}$ and $\Nspcenorm{U}{2j+1} =
  \Nspcenorm{U}{2j}$ whenever $2j+1 \leq n$.
 \item If $U(t)$ is a continuous one-parameter family of unitary matrices,
  then $\enorm{U(t)}$ and $\Nenorm{U(t)}$ are continuous.  This result is more
  involved.  By the theorem that roots of a polynomial vary continuously with
  its coefficients~\cite{polynomial_roots,polynomial_book}, eigenvalues of
  $U(t)$ must vary continuously on the unit circle.  And as the absolute value
  of argument of a complex-valued and nowhere zero function is continuous,
  $\enorm{U(t)}$ is continuous.  The continuity of $\Nenorm{U(t)}$ then follows
  the fact that the pointwise minimum of a collection of continuous functions
  is also a continuous function if it exists.
 \item If $H(t)$ is an $n$-dimensional Hamiltonian and $U(t)$ is the unitary
  operator generated by $H(t)$, then
  \begin{equation}
   \left. \frac{d \enorm{U(t)}}{dt} \right|_{t=0} = \sum_{j=1}^n \mu_j
   \desabseig{H(0)}{j} .
   \label{E:denormdt}
  \end{equation}
  This shows the relation between rate of change of $\enorm{U(t)}$ and the
  singular values of its generator $H(t)$.  Actually, the R.H.S. of
  Eq.~(\ref{E:denormdt}) is the so-called generalized spectral
  norm~\cite{generalized_Ky_Fan_norm} of the Hamiltonian $H(0)$.
 \item $\emetric{\cdot}{\cdot}$ and $\Nemetric{\cdot}{\cdot}$ are positive
  definite and positive semi-definite functions, respectively.  And they are
  both symmetric functions.
 \item $\emetric{U}{V} = \pi \sum_j \mu_j$ for all $\vec{\mu}$ if and only if
  $V = -U$.
 \item Let $U_i$ be $n_i$-dimensional unitary matrices, then $\spcenorm{U_1
  \otimes U_2}{1} \leq \spcenorm{U_1}{1} + \spcenorm{U_2}{1}$, $\Nspcenorm{U_1
  \otimes U_2}{1} \leq \Nspcenorm{U_1}{1} + \Nspcenorm{U_2}{1}$, $\spcenorm{U_1
  \otimes U_2}{n_1 n_2} \leq n_2 \spcenorm{U_1}{n_1} + n_1 \spcenorm{U_2}{n_2}$
  and $\Nspcenorm{U_1 \otimes U_2}{n_1 n_2} \leq n_2 \Nspcenorm{U_1}{n_1} + n_1
  \Nspcenorm{U_2}{n_2}$.  (Note that I have abused the notation a little bit as
  clearly the lengths of the vectors $\vecmu{1},$ $\vecmu{n_1},$ $\vecmu{n_2},$
  $\vecmu{n_1 n_2}$ in the above inequalities are different even though each of
  them has the same number of non-zero entries.)
\end{itemize}

 Interestingly, except for the second property (the necessary and sufficient
 condition for $\enorm{\cdot}$ and $\Nenorm{\cdot}$ to be zero) and the last
 property (concerning the tensor product of unitary matrices), all the above
 properties require only $\mu_j \geq 0$ for all $j$.  In contrast, the
 condition $\mu_1 \geq \mu_2 \geq \cdots \geq \mu_n \geq 0$ is required for
 $\emetric{\cdot}{\cdot}$ and $\emetric{\cdot}{\cdot}$ to satisfy the triangle
 inequality.  In fact, the triangle inequality follows directly from
 Eq.~(\ref{E:metric_def}) and the theorem below.  And this theorem alone is
 also of great interest as it bounds certain weighted sums of the arguments of
 eigenvalues for the product of two unitary matrices.

\begin{theorem}
 Let $U, V$ be two $n$-dimensional unitary matrices and $\vec{\mu}$ satisfying
 the requirements stated in Definition~\ref{Def:norm_metric}.  Then
 \begin{equation}
  \left| \enorm{U} - \enorm{V} \right| \leq \enorm{U V} \leq \enorm{U} +
  \enorm{V} .
  \label{E:UV_inequality}
 \end{equation}
 \label{Thrm:angle_bound}
\end{theorem}

 I move both the proof of Theorem~\ref{Thrm:angle_bound} and the conditions for
 the second inequality in Eq.~(\ref{E:UV_inequality}) to become an equality to
 Appendix~\ref{Sec:proof_of_angle_bound} as they are rather involved and
 technical.

\begin{corollary}
 Eq.~(\ref{E:UV_inequality}) in Theorem~\ref{Thrm:angle_bound} also holds if
 $\enorm{\cdot}$ is replaced by $\Nenorm{\cdot}$.
 \label{Cor:angle_bound}
\end{corollary}
\begin{proof}
 For any $n$-dimensional unitary matrices $U$ and $V$, there exist $x, y \in
 {\mathbb R}$ such that $\enorm{e^{i x} U} = \Nenorm{U}$ and $\enorm{e^{i y}
 V} = \Nenorm{V}$.  By Theorem~\ref{Thrm:angle_bound}, $\Nenorm{U} + \Nenorm{V}
 \geq \enorm{e^{i (x+y)} U V} \geq \Nenorm{U V}$.  The proof of $\left|
 \Nenorm{U} - \Nenorm{V} \right| \leq \Nenorm{U V}$ is then a replica of that
 of the first half of Eq.~(\ref{E:UV_inequality}).
\end{proof}

 From Theorem~\ref{Thrm:angle_bound} and Corollary~\ref{Cor:angle_bound}, I
 conclude that $\emetric{\cdot}{\cdot}$ and $\Nemetric{\cdot}{\cdot}$ are a
 metric and pseudo-metric, respectively.  Furthermore, $\enorm{\cdot},
 \Nenorm{\cdot}$ behave somewhat like a norm and semi-norm respectively in the
 sense that $\enorm{U V} \leq \enorm{U} + \enorm{V}$.  In other words, the
 addition operation in the conventional definition of a norm is replaced by
 multiplication.  Note further that $\enorm{\cdot}$ does not obey $\enorm{a U}
 = |a| \enorm{U}$ for all scalars $a$.

\begin{figure*}[t]
 \centering
 \includegraphics*[scale=0.32]{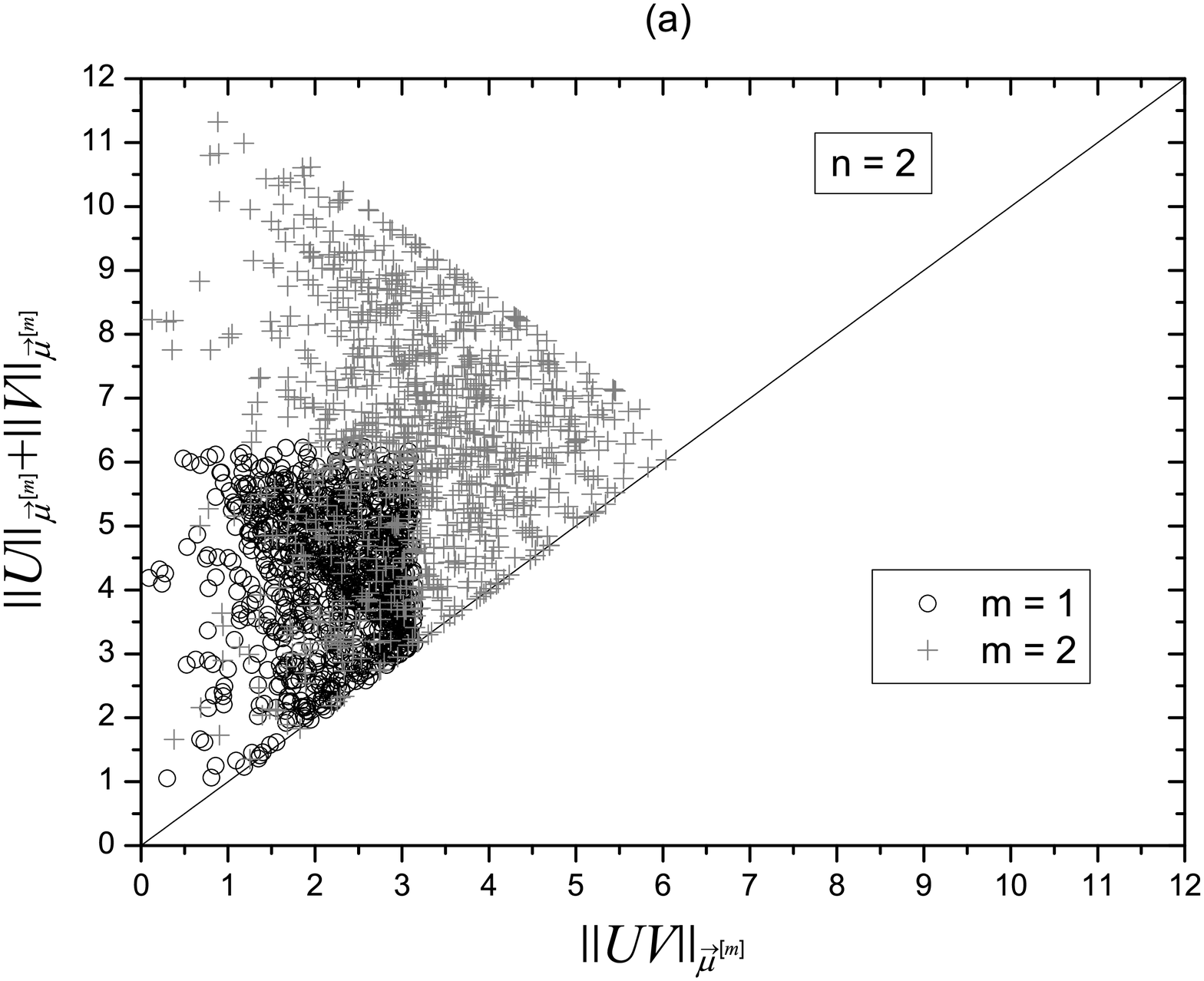}
 \includegraphics*[scale=0.32]{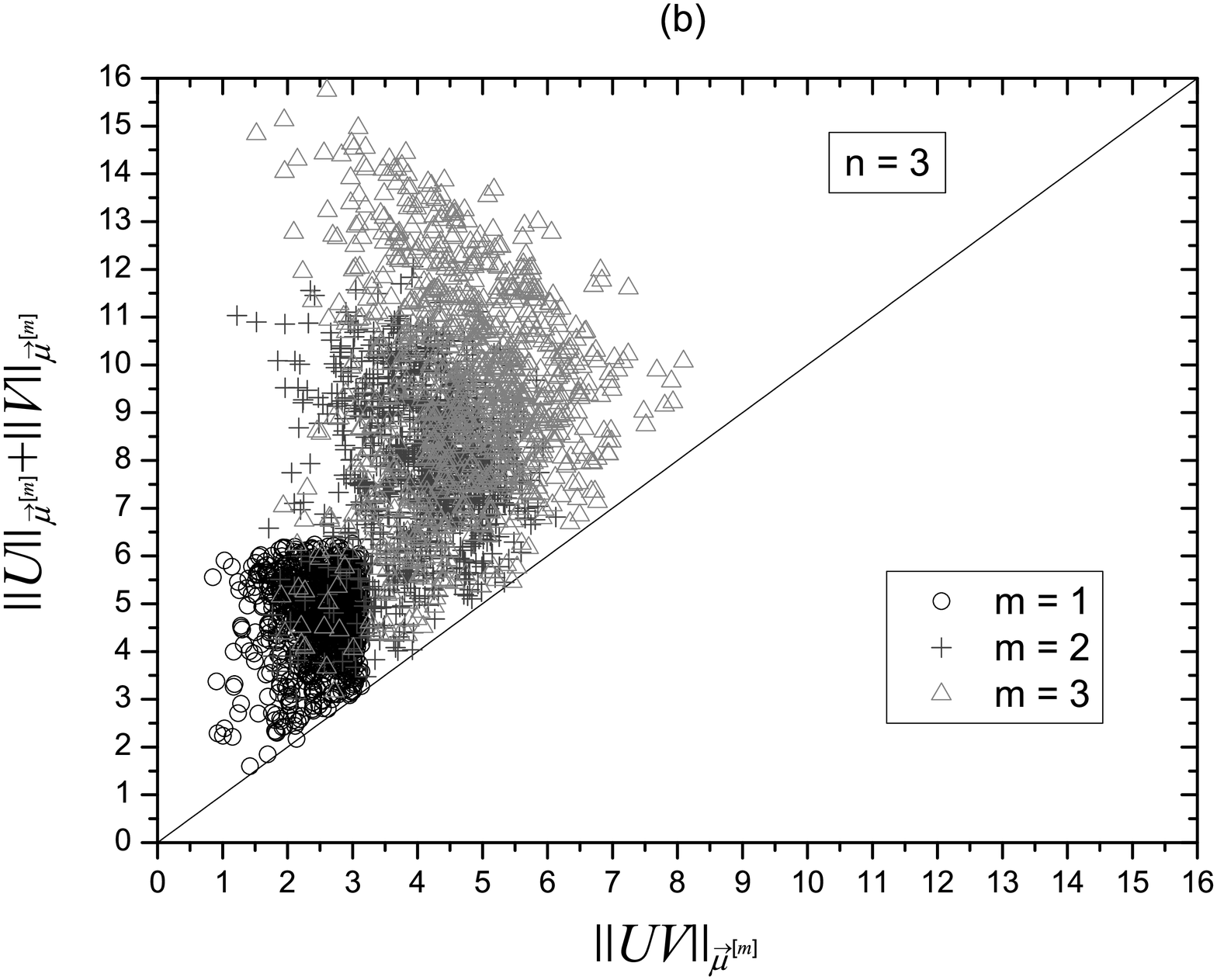}
 \includegraphics*[scale=0.32]{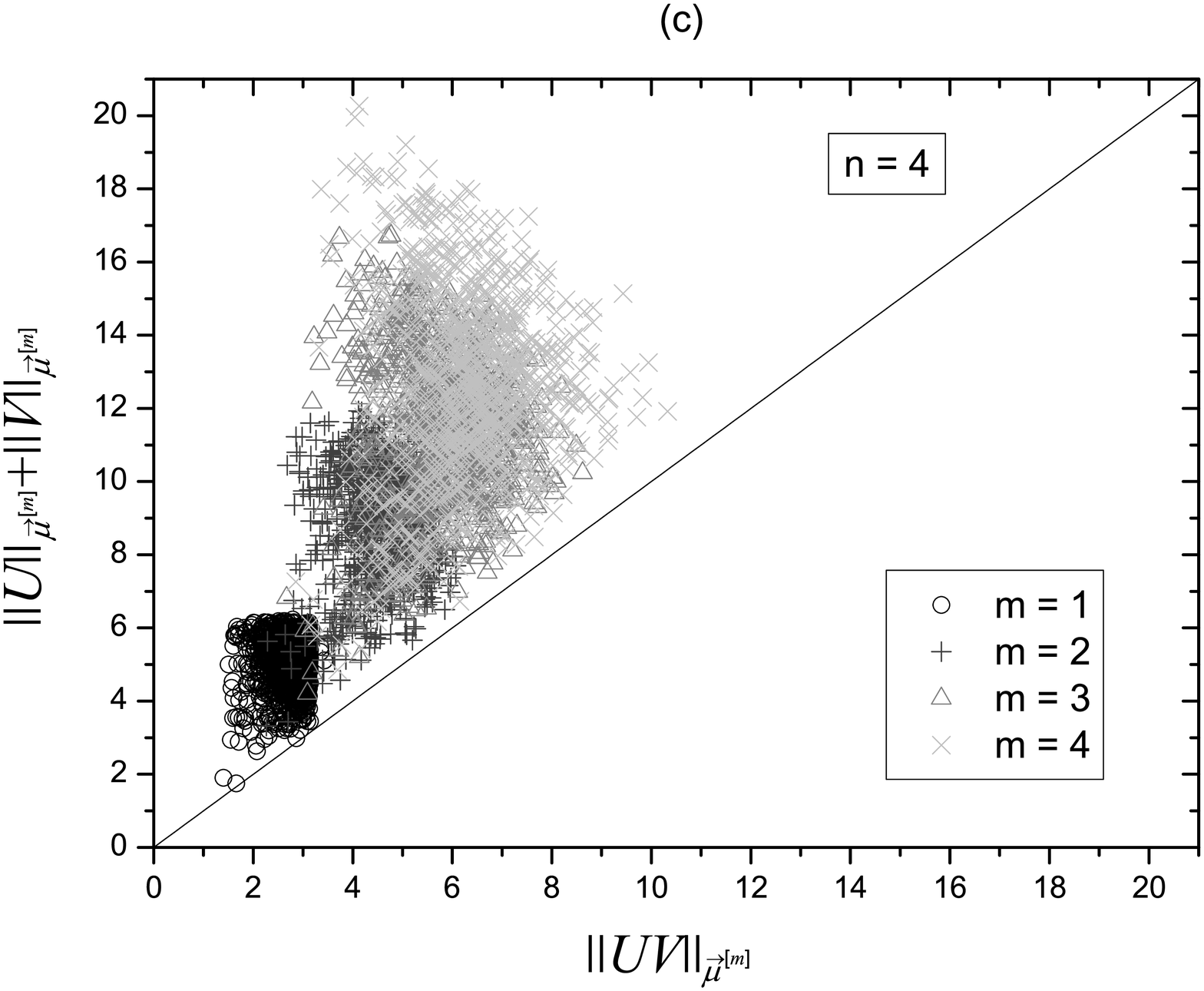}
 \caption{Plots of $\spcenorm{U V}{m}$ against $\spcenorm{U}{m} +
  \spcenorm{V}{m}$ for various values of $m$ where the dimension $n$ of the
  unitary matrices $U$ and $V$ equals (a)~$n = 2$, (b)~$n = 3$ and (c)~$n = 4$.
  The matrices $U$ and $V$ are randomly chosen from the Haar measure of $U(n)$
  and the sample size in each case is $1000$.  The black lines represent the
  situation in which $\spcenorm{U V}{m} = \spcenorm{U}{m} + \spcenorm{V}{m}$.
  \label{F:enorm}
 }
\end{figure*}

\begin{figure*}[t]
 \centering
 \includegraphics*[scale=0.32]{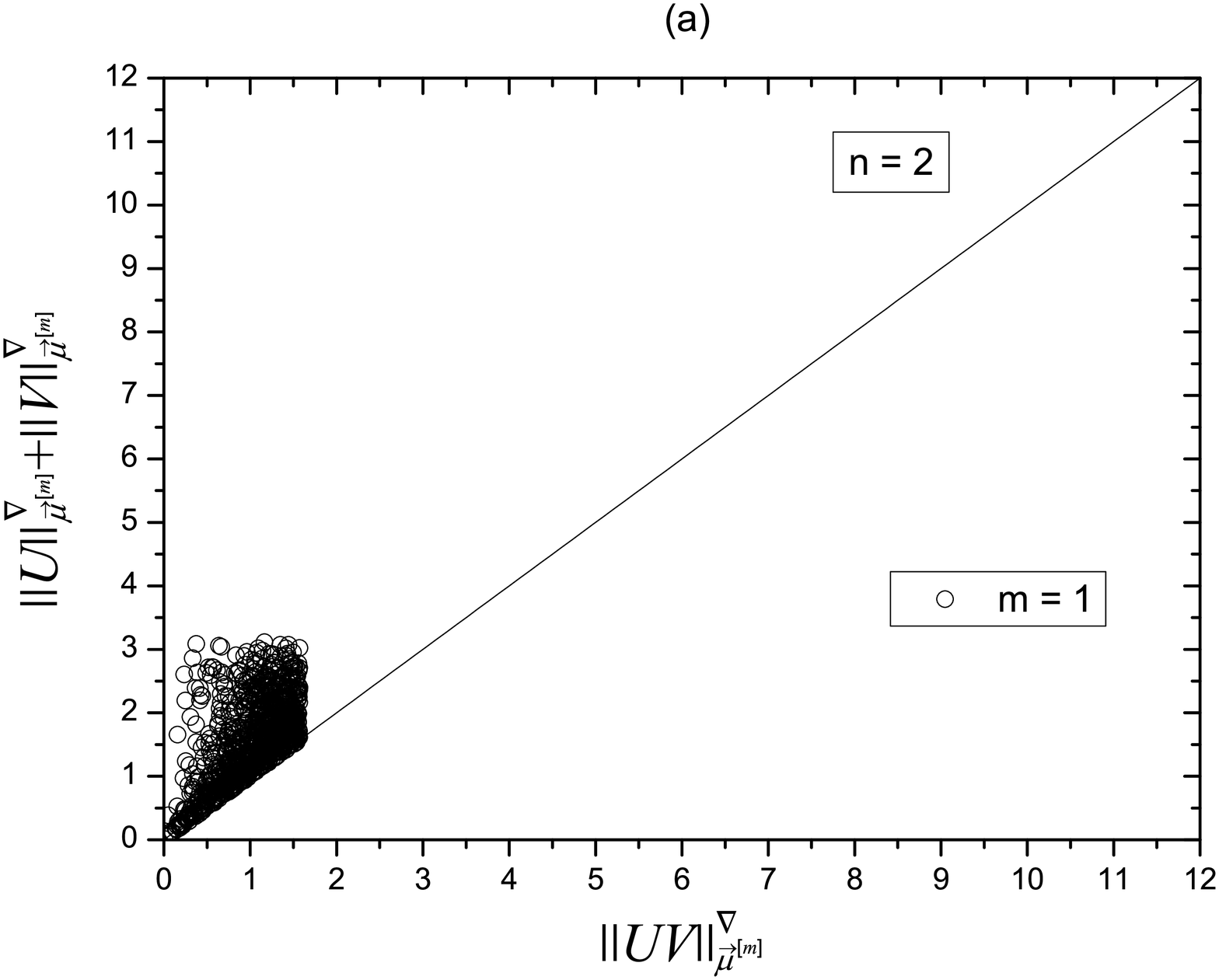}
 \includegraphics*[scale=0.32]{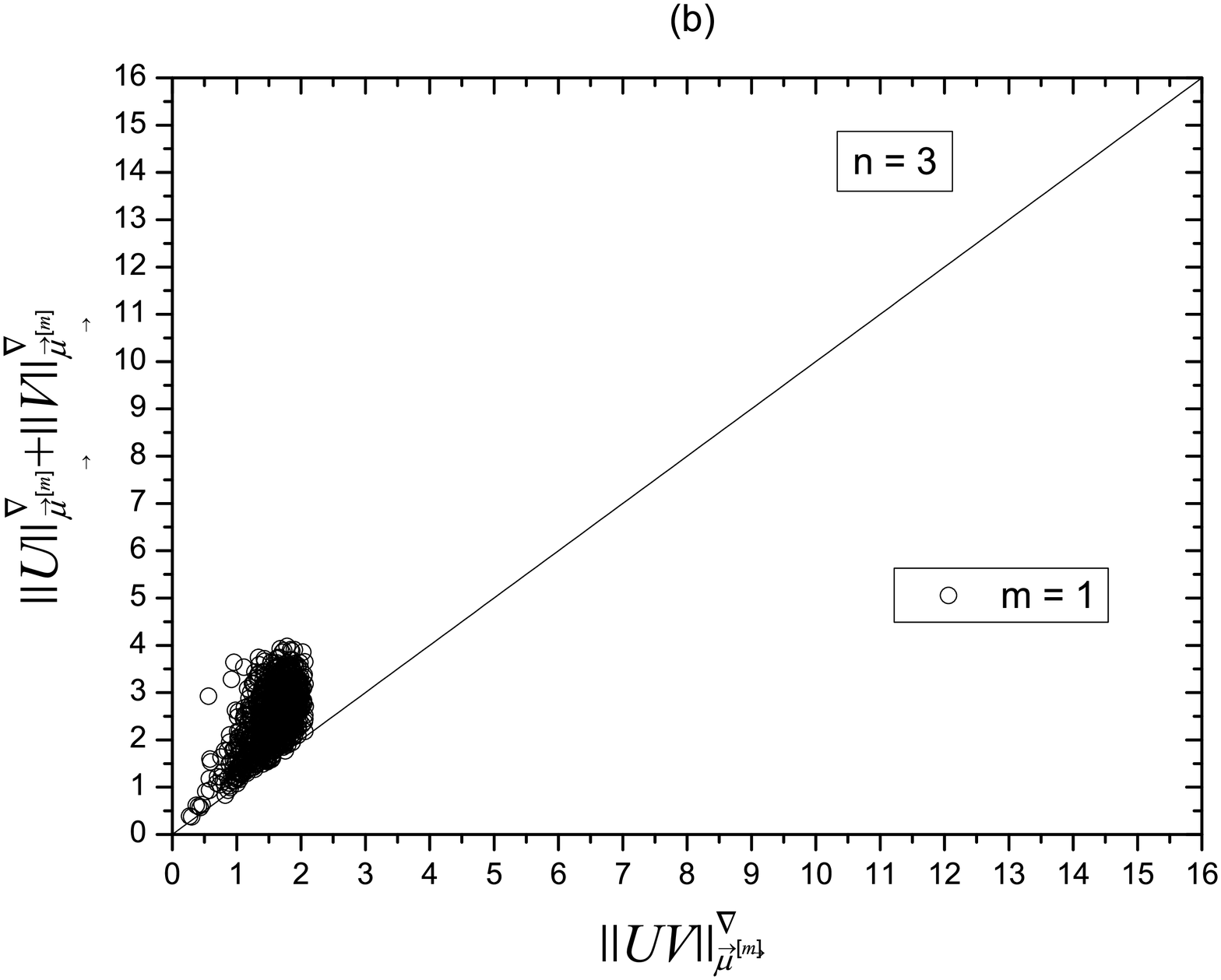}
 \includegraphics*[scale=0.32]{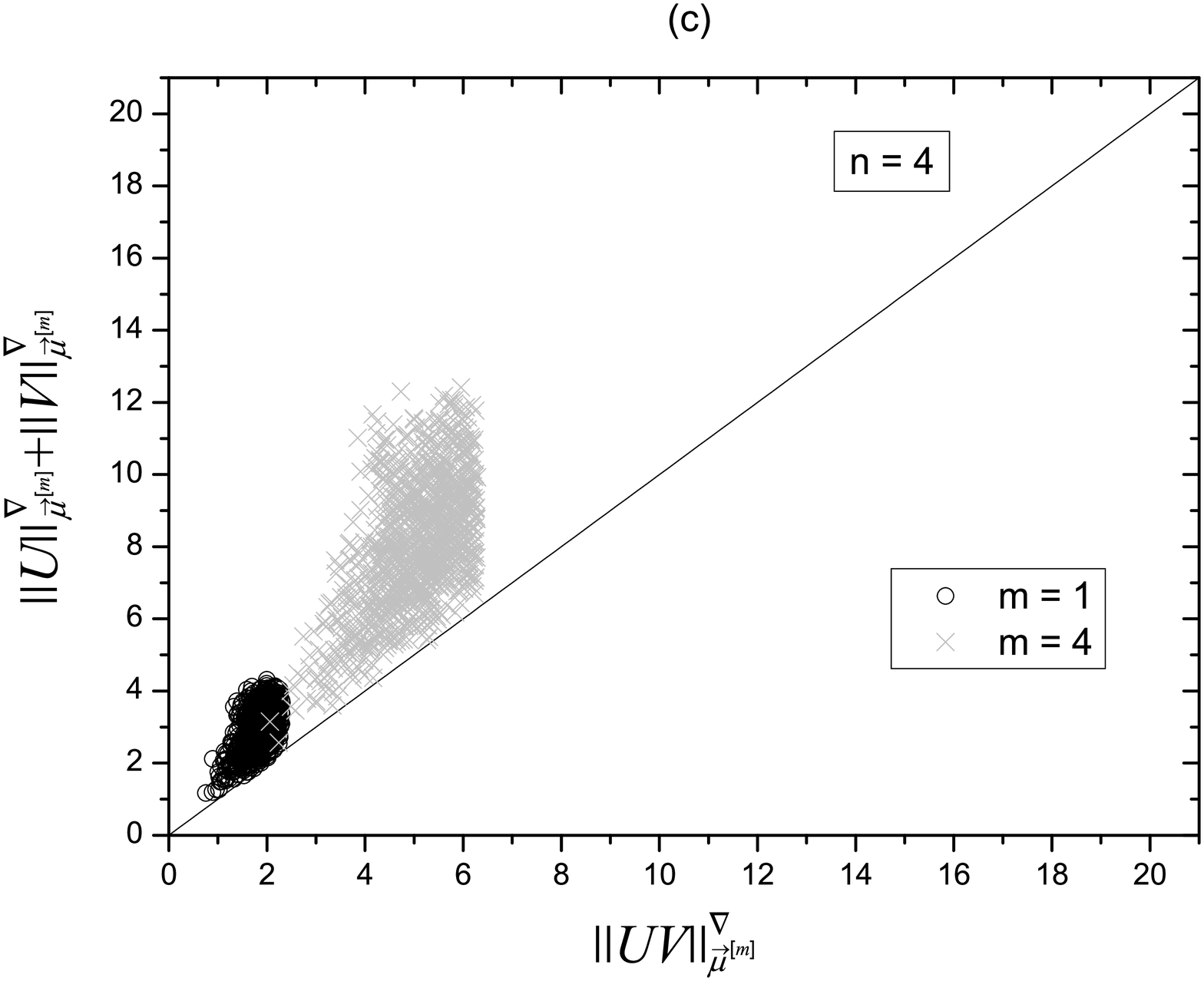}
 \caption{Plots of $\Nspcenorm{U V}{m}$ against $\Nspcenorm{U}{m} +
  \Nspcenorm{V}{m}$ for various values of $m$ where the dimension $n$ of the
  unitary matrices $U$ and $V$ equals (a)~$n = 2$, (b)~$n = 3$ and (c)~$n = 4$.
  The parameters and sampling method used are the same as those in
  Fig.~\ref{F:enorm}.  Since $\Nspcenorm{U}{3} = \Nspcenorm{U}{2} = 2
  \Nspcenorm{U}{1}$, I only show the cases of $m = 1$ and $m = 4$ in this
  figure.  And to allow easy comparison, the scales used in the plots are the
  same as those used in Fig.~\ref{F:enorm}.
  \label{F:Nenorm}
 }
\end{figure*}

\begin{remark}
 For any $n$-dimensional matrices $U$ and $V$ satisfying $\desangle{U}{1} +
 \desangle{V}{1} \leq \pi$ and $\desangle{U}{n} + \desangle{V}{n} > -\pi$, the
 following closely related bound
 \begin{equation}
  \desangle{U V}{1} \leq \desangle{U}{1} + \desangle{V}{1}
  \label{E:related_bound1}
 \end{equation}
 has been proven by a few authors using different
 methods~\cite{convex_hull_bound,q_calculus,principal_bound}.  Recently, our
 group~\cite{lam_thesis,elementary_proof} also found elementary proofs of
 Eq.~(\ref{E:related_bound1}) as well as the special case of
 Theorem~\ref{Thrm:angle_bound} in which $\vec{\mu} = \vecmu{1}$.
 \label{Rem:history}
\end{remark}

\begin{remark}
 The proof for the special case of $\vec{\mu} = \vecmu{n}$ is rather simple.  I
 simply need to use the identity $\prod_{j=1}^n e^{i\desangle{U V}{j} / 2n} =
 \det [ (U V)^{1/2n} ] = [ \det ( U V ) ] ^{1/2n} = \left[ \det U \det V
 \right]^{1/2n} = \det U^{1/2n} \det V^{1/2n} = \prod_{j=1}^n e^{i \left[
 \desangle{U}{j} + \desangle{V}{j} \right] / 2n}$.  Then, by equating the
 arguments on both sides and by observing that both arguments are in the
 interval $(-\pi,\pi]$, I obtain the required inequality.
 \label{Rem:determinant}
\end{remark}

\begin{remark}
 It is easy to check from the proof of Corollary~\ref{Cor:angle_bound} that
 $\Nenorm{U V} = \Nenorm{U} + \Nenorm{V}$ if and only if there exist $x, y \in
 {\mathbb R}$ such that $\Nenorm{U} = \enorm{e^{i x} U}$, $\Nenorm{V} =
 \enorm{e^{i y} V}$, $\enorm{e^{i (x+y)} U V} = \enorm{e^{i x} U} +
 \enorm{e^{i y} V}$ and $\Nenorm{U V} = \enorm{e^{i (x+y)} U V}$.  Therefore,
 the conditions for equality of the triangle inequality for $\Nenorm{\cdot}$
 seems to be more stringent than those for $\enorm{\cdot}$.  On the other hand,
 as I have pointed out in the list of the basic properties of $\Nenorm{\cdot}$
 that $\Nspcenorm{U}{3} = \Nspcenorm{U}{2} = 2\Nspcenorm{U}{1}$.  Thus, the
 conditions for $\Nspcenorm{U V}{m} = \Nspcenorm{U}{m} + \Nspcenorm{V}{m}$ is
 not particularly more stringent than those for $\spcenorm{\cdot}{m}$ for $m
 \leq 3$.  But in any case, as $m$ or $n$ increases, it is more and more
 difficult for $\Nenorm{U V} = \Nenorm{U} + \Nenorm{V}$ provided that $U$ and
 $V$ are drawn randomly from the Haar measure of $U(n)$.  This finding is
 reflected in the plots in Fig.~\ref{F:Nenorm}.  Interestingly, by
 comparing Figs.~\ref{F:enorm} and~\ref{F:Nenorm}, I find that on average (over
 the Haar measure of $U(n)$) the value of $\Nenorm{U} + \Nenorm{V} -
 \Nenorm{U V}$ is smaller than $\enorm{U} + \enorm{V} - \enorm{U V}$.  This is
 probably due to the fact that $\Nenorm{U} \leq \enorm{U}$.
 \label{Rem:Nenorm}
\end{remark}

 Remarkably, several new inequalities involving the eigenvalues of a product of
 two unitary matrices can be found by modifying the proof of
 Theorem~\ref{Thrm:angle_bound}~\cite{original_article}.  But I will not
 further elaborate on this matter here as it is beyond the scope of this paper.

\section{Measure Of The Degree Of Non-Commutativity Between Two Unitary
 Operators \label{Sec:Non-Commutativity}}
 Although the family of $\enorm{\cdot}$'s are only ``unoptimized'' measures of
 the minimum resources needed to perform the unitary transformation in its
 argument, it is useful in measuring the degree of non-commutativity between
 two unitary operators.

\begin{definition}
 Let $U, V$ two $n$-dimensional unitary matrices.  Suppose $\vec{\mu}$ be the
 vector defined in Def.~\ref{Def:norm_metric}.  I define
 \begin{equation}
  \ecomm{U}{V} = \emetric{U V}{V U} \equiv \enorm{U V U^{-1} V^{-1}} .
  \label{E:commutator_norm}
 \end{equation}
 \label{Def:non_commutativity}
\end{definition}

 Physically, $\ecomm{U}{V}$ can be interpreted as a measure of the minimum
 average absolute deviation from the median energy of the system times the
 evolution time required to convert $U V$ to $V U$.  The higher the required
 resources, the more ``non-commutative'' the operators $U$ and $V$ are.

 The following simple properties of $\ecomm{\cdot}{\cdot}$ carry over from the
 properties of $\enorm{\cdot}$ reported in Sec.~\ref{Sec:Properties}:
\begin{itemize}
 \item $\ecomm{U}{V}: U(n)\times U(n) \longrightarrow [0,\pi\sum_{j=1}^n
  \mu_j]$ is a surjective map.
 \item $\ecomm{U}{V} = 0$ if and only if $U$ commutes with $V$.
 \item $\ecomm{\cdot}{\cdot}$ is not a metric.  In fact, any meaningful measure
  of commutativity is not a metric for commutativity is not transitive.
 \item $\ecomm{e^{i x}U}{e^{i y}V} = \ecomm{U}{V}$ for any $x,y\in
  {\mathbb R}$.
 \item $\ecomm{U^{-1}}{V^{-1}} = \ecomm{U}{V} =
  \ecomm{W U W^{-1}}{W V W^{-1}}$.
 \item ${\mathfrak C}_{a\vec{\mu}}(U,V) = a\ecomm{U}{V}$ for all $a > 0$.
 \item $\ecomm{U}{V} = \pi \sum_j \mu_j$ if $U V U^{-1} V^{-1} = -I$.  In
  particular, the three Pauli matrices $\sigma_x$, $\sigma_y$ and $\sigma_z$
  obey $\ecomm{\sigma_x}{\sigma_y} = \ecomm{\sigma_y}{\sigma_z} =
  \ecomm{\sigma_z}{\sigma_x} = \pi \sum_k \mu_k$.  In this sense, pairs of
  distinct Pauli matrices are examples of the most non-commutative unitary
  operators in $U(2)$.
 \item ${\mathfrak C}_{\vecmu{1}}(U_1\otimes V_1,U_2\otimes V_2) \leq
  {\mathfrak C}_{\vecmu{1}}(U_1,V_1) + {\mathfrak C}_{\vecmu{1}}(U_2,V_2)$, and
  ${\mathfrak C}_{\vecmu{n_1 n_2}}(U_1\otimes V_1,U_2\otimes V_2) \leq n_2
  {\mathfrak C}_{\vecmu{n_1}}(U_1,V_1) + n_1 {\mathfrak C}_{\vecmu{n_2}}
  (U_2,V_2)$.
 \item If $U(t)$ and $V(t)$ are continuous one-parameter families of unitary
  matrices, then $\ecomm{U(t)}{V(t)}$ is continuous.
 \item Let $H_i(t)$ be an $n$-dimensional Hamiltonian and $U_i(t)$ be the
  unitary operator generated by $H_i(t)$ for $i=1,2$.  Then,
  \begin{eqnarray}
   & & \left. \frac{d^2 \ecomm{U_1(t)}{U_2(t)}}{dt^2} \right|_{t=0} \nonumber
    \\
   & = & 2 \lim_{t\rightarrow 0} \frac{\ecomm{U_1(t)}{U_2(t)}}{t^2} \nonumber
    \\
   & = & 2\sum_{j=1}^n \mu_j \desabseig{-i [ H_1(0), H_2(0)]}{j} .
   \label{E:d2ecommdt2}
  \end{eqnarray}
  This shows the relation between the curvature of $\ecomm{U_1(t)}{U_2(t)}$ and
  the singular values of the commutator of the the corresponding generators
  $[H_1(t), H_2(t)]$.
\end{itemize}
 
\section{Discussions \label{Sec:Discuss}}
 In summary, I have introduced a family of metrics $\emetric{\cdot}{\cdot}$ and
 a family of pseudo-metrics $\Nemetric{\cdot}{\cdot}$ on finite-dimensional
 unitary matrices.  In particular, the pseudo-metric $\Nemetric{U}{V}$ can be
 interpreted as a measure of the minimum average absolute deviation from the
 median energy of the system times the evolution time needed to perform the
 unitary transformation $U V^{-1}$ (and hence equivalently also $V U^{-1}$).
 Besides, $\Nemetric{\cdot}{\cdot}$ is related to the Bures angle between two
 quantum states; while $\emetric{U}{V}$ is related to the generalized spectral
 norm of the infinitesimal generator of $U V^{-1}$.

 Another aspect of this paper is the proposed measure on the degree of
 non-commutativity between two unitary matrices $U$ and $V$ based on
 $\enorm{U V U^{-1} V^{-1}}$.  Its physical meaning is a measure of the minimum
 resources needed to convert $U V$ to $V U$.  Interestingly, this measure is
 related to the generalized spectral norm of the commutator of the
 infinitesimal generators of $U$ and $V$.

 The analysis here so far are restricted to finite-dimensional unitary
 matrices.  It is instructive to see how it can be extended to cover the
 infinite-dimensional case.  Also a possible future research direction is to
 investigate the possibility of extending the results here to trace-preserving
 completely positive maps and to find non-trivial applications in quantum
 information processing.

\appendix
\section{Proof of Theorem~\ref{Thrm:interpretation}}
\label{Sec:proof_of_interpretation}
 The following lemma is needed to prove Theorem~\ref{Thrm:interpretation}.

\begin{lemma}
 Let $x_j$'s and $y_j$'s be real numbers satisfying $x_1 \geq x_2 \geq \cdots
 \geq x_n$ and $y_1 \geq y_2 \geq \cdots \geq y_n$.  Then
 \begin{equation}
  \sum_{j=1}^n x_j y_{P(j)} \leq \sum_{j=1}^n x_j y_j \label{E:dot_product}
 \end{equation}
 for all permutation $P$ of the set $\{ 1,2,\cdots ,n \}$.
 \label{Lem:dot_product}
\end{lemma}
\begin{proof}
 Eq.~(\ref{E:dot_product}) is trivially true for $n = 1$.  And its validity for
 $n = 2$ follows from the inequality $(x_1 - x_2)(y_1 - y_2) \geq 0$.  This
 Lemma can then be proven by mathematical induction on $n$.
\end{proof}

\begin{proof}[Proof of Theorem~\ref{Thrm:interpretation}]
 Note that $t$ can only be $0$ when $e^{i x} U = I$.  In this case,
 Eq.~(\ref{E:min_max_interpretation}) clearly makes sense and is equal to $0$.
 So, this theorem is trivially true.  Whereas if $t < 0$, then
 ${\mathscr D}E(H,|\phi\rangle) t \leq 0$.  In this case, I may consider $-H$
 and $-t$ instead in analyzing Eq.~(\ref{E:min_max_interpretation}) for
 ${\mathscr D}E(-H,|\phi\rangle) \,(-t) \geq 0$ and $\exp ( - i H t / \hbar ) =
 \exp [ -i (-H) (-t) / \hbar ]$.  Hence, I only need to consider the remaining
 the case of $t > 0$ from now on.

 I first justify the use of maximum and minimum in
 Eq.~(\ref{E:min_max_interpretation}); and I do so by considering a similar
 min-max expression:
 \begin{widetext}
 \begin{equation}
  R' = \min_{x\in [0,2\pi)} \quad \min_{H t \colon \exp ( - i H t / \hbar) =
  e^{i x} U} \quad \max_{|\phi\rangle \in C(H,(\alpha_j))} \frac{A \langle\phi|
  \sqrt{H^\dag H} |\phi\rangle t}{\hbar} .
  \label{E:alt_min_max_interpretation}
 \end{equation}
 \end{widetext}
 Note that the set $\{ \langle\phi| \sqrt{H^\dag H} |\phi\rangle \colon
 |\phi\rangle\in C(H,(\alpha_j)) \}$ equals $\{ \sum_{j=1}^n |\alpha_j|^2
 \desabseig{H}{P(j)} \colon P\in S_n \}$, where $S_n$ denotes the permutation
 group of $n$ elements.  So from Lemma~\ref{Lem:dot_product}, for any given
 $H$, $t$ and $x$, the maximum in Eq.~(\ref{E:alt_min_max_interpretation})
 exists and is equal to $A t \sum_{j=1}^n |\alpha_j|^2 \desabseig{H}{j} /
 \hbar$.  Now suppose $\desabseig{H}{1} \geq \pi\hbar / t$.  I consider the new
 Hamiltonian $H'$ formed by changing only the eigenvalue $\lambda$
 corresponding to $\desabseig{H}{1}$ to $\lambda \bmod (2\pi\hbar / t) \equiv
 \lambda' \in (-\pi\hbar / t,\pi\hbar/t]$.  Clearly, $|\lambda'| \leq |\lambda|
 = \desabseig{H}{1}$ and $\exp \left( -i H' t / \hbar \right) = \exp \left( -i
 H t / \hbar \right)$.  In other words, there exists $k \in \{ 1,2,\ldots ,n\}$
 such that
 \begin{equation}
  \desabseig{H'}{j} = \left\{ \begin{array}{ll}
   \desabseig{H}{j+1} & \text{if\ } j < k , \\
   |\lambda'| & \text{if\ } j = k, \\
   \desabseig{H}{j} & \text{if\ } j > k .
  \end{array} \right.
  \label{E:repermutation}
 \end{equation}
 Consequently, by Lemma~\ref{Lem:dot_product},
 \begin{widetext}
 \begin{eqnarray}
  \max_{|\phi'\rangle\in C(H',(\alpha_j))} \langle\phi'| \sqrt{H'^\dag H'}
   |\phi'\rangle - \max_{|\phi\rangle\in C(H,(\alpha_j))} \langle\phi|
   \sqrt{H^\dag H} |\phi\rangle
  & = & |\alpha_k|^2 |\lambda'| - |\alpha_1|^2 \desabseig{H}{1} +
   \sum_{\ell=1}^{k-1} \left( |\alpha_\ell|^2 - |\alpha_{\ell+1}|^2 \right)
   \desabseig{H}{\ell+1} \nonumber \\
  & \leq & \desabseig{H}{1} \left[ |\alpha_k|^2 - |\alpha_1|^2 +
   \sum_{\ell=1}^{k-1} \left( |\alpha_\ell|^2 - |\alpha_{\ell+1}|^2 \right)
   \right] \nonumber \\
  & = & 0 .
 \end{eqnarray}
 \end{widetext}
 That is to say, $\max_{|\phi\rangle\in C(H,(\alpha_j))} \langle\phi|
 \sqrt{H^\dag H} |\phi\rangle$ does not increase if $H$ is replaced by $H'$.
 By repeating this procedure at most $n$ times, I conclude that in order to
 look for the minimum in Eq.~(\ref{E:alt_min_max_interpretation}), I only need
 to consider those $H$'s with all their eigenvalues in $(-\pi\hbar/t,
 \pi\hbar/t]$.  Note that the expression $\max_{|\phi\rangle\in C(H,(
 \alpha_j))} \langle\phi| \sqrt{H^\dag H} |\phi\rangle$ depends on the
 eigenvalues of $H$ rather than the eigenvectors of $H$.  Moreover, for a fixed
 $x\in {\mathbb R}$, there is a unique set of eigenvalues $\deseig{H}{j}$'s for
 $H$ such that $\deseig{H}{j} \in (-\pi\hbar / t,\pi\hbar / t]$ and $\exp
 \left( -i H t / \hbar \right) = e^{i x} U$.  So, the second minimum expression
 (that is, the minimum over $H t$) in Eq.~(\ref{E:alt_min_max_interpretation})
 exists for each $x\in {\mathbb R}$.  As $x$ varies, the set of eigenvalues of
 $H$ that minimizes the second minimum expression in
 Eq.~(\ref{E:alt_min_max_interpretation}) changes linearly with $x$ modulo $2
 \pi\hbar / t$.  Hence, this second minimum expression is a continuous function
 of $x\in {\mathbb R}$ with period $2\pi$.  Therefore, $R'$ exists.

 Now I go on to show that the $H$ and $|\phi\rangle$ that minimize
 Eq.~(\ref{E:alt_min_max_interpretation}) can be chosen to have median
 system energy $M = 0$.  First, I claim that $\sum^- |\alpha_j|^2 \geq 1/2$
 where the sum is over all $j$'s with $\deseig{H}{j} t / \hbar \in \{\pi\} \cup
 (-\pi,0]$.  Suppose the contrary, as $H$ has at most $n$ distinct
 eigenvalues, there exits $\epsilon < 0$ sufficiently close to 0 such that
 $\langle\phi| \sqrt{\left( H + \epsilon \right)^\dag \left( H + \epsilon
 \right)} |\phi\rangle < \langle\phi| \sqrt{H^\dag H} |\phi\rangle$.  This
 contradicts the assumption that $H$ and $|\phi\rangle$ extremize
 Eq.~(\ref{E:alt_min_max_interpretation}).  By a similar argument, $\sum^+
 |\alpha_j|^2 \leq 1/2$ where the sum is over all $j$'s with $\deseig{H}{j} t /
 \hbar \in [0,\pi)$.  Hence, from Eqs.~(\ref{E:Median_def1})
 and~(\ref{E:Median_def2}), the median energy of the system $M$ is equal to
 $0$.

 From the discussions between Eqs.~(\ref{E:absolute1}) and~(\ref{E:absolute2})
 on the minimization of $\sum_j |\alpha_j|^2 |\deseig{H'}{j} - x|$, I conclude
 that ${\mathscr D}E(H',|\phi\rangle) \leq \max_{|\phi\rangle\in C(H',
 (\alpha_j))} \langle\phi| \sqrt{H'^\dag H'} |\phi\rangle$ with equality holds
 if the median system energy $M = 0$.  By comparing
 Eq.~(\ref{E:min_max_interpretation}) with
 Eq.~(\ref{E:alt_min_max_interpretation}), I deduce that $R \leq R'$ if $R$
 exists.  Recall that Eq.~(\ref{E:alt_min_max_interpretation}) is well-defined
 and the extremum can be attained by $H$ and $|\phi\rangle$ such that
 $\deseig{H}{j} \in (-\pi\hbar / t,\pi\hbar / t]$ and the median system energy
 $M = 0$.  With these $H$ and $|\phi\rangle$, ${\mathscr D}E(H,|\phi\rangle) =
 \sum_{j=1}^n |\alpha_j|^2 \desabseig{H}{j} = \langle\phi| \sqrt{H^\dag H}
 |\phi\rangle$.  Hence, $R$ exists and is equal to $R'$.  That is to say,
 Eq.~(\ref{E:min_max_interpretation}) is well-defined and its extremum is
 attained by picking $H$ and $|\phi\rangle$ so that $M = 0$ and $\deseig{H}{j}
 t / \hbar \in (-\pi,\pi]$.
\end{proof}

\section{Proof of Theorem~\ref{Thrm:angle_bound} and the conditions for
 equality in Eq.~(\ref{E:UV_inequality})}
\label{Sec:proof_of_angle_bound}

\begin{proof}[Proof of Theorem~\ref{Thrm:angle_bound}]
 I only need to show the second half of the inequality in
 Eq.~(\ref{E:UV_inequality}) as the first half follows directly from it.  More
 precisely, from the second half of Eq.~(\ref{E:UV_inequality}), $\enorm{U}
 \leq \enorm{U V} + \enorm{V^{-1}} = \enorm{U V} + \enorm{V}$; and similarly
 $\enorm{V} \leq \enorm{U^{-1}} + \enorm{U V} = \enorm{U} + \enorm{U V}$.  And
 from Eq.~(\ref{E:enorm_decomposition}), it suffices to prove this theorem by
 showing its validity for all $\vecmu{j}$ $(j=1,2,\ldots , n)$.

 Let $\epsilon > 0$ be a small real number.  Then,
 \begin{equation}
  U V^\epsilon = \sum_{j,\ell} e^{i \left( \unorderangle{U}{j} + \epsilon
  \unorderangle{V}{\ell} \right)} \unorderbarebra{U}{j}\unorderket{V}{\ell}
  \,\,\unorderket{U}{j}\unorderbra{V}{\ell} .
  \label{E:UVepsilon}
 \end{equation}
 I now follow the strategy used in Ref.~\cite{na_book} to bound the eigenvalues
 of $U V^\epsilon$.  By writing $U V^\epsilon$ in the orthonormal basis $\{
 \unorderket{U}{j} \}_{j=1}^n$, I can regard $U V^\epsilon$ as a matrix with
 matrix elements
 \begin{equation}
  \left( U V^\epsilon \right)_{jk} = e^{i \unorderangle{U}{j}} \sum_\ell e^{i
  \epsilon \unorderangle{V}{\ell}} \unorderbarebra{U}{j}\unorderket{V}{\ell}
  \,\unorderbarebra{V}{\ell}\unorderket{U}{k} . \label{E:UVepsilon_elements}
 \end{equation}

 Let me consider the effect of $V^\epsilon$ on the non-degenerate eigenvalues
 of $U$ first.  Suppose $e^{i \unorderangle{U}{a}}$ is a non-degenerate
 eigenvalue of the unperturbed unitary matrix $U$.  I define an invertible
 matrix $F$ by
 \begin{equation}
  F_{jk} = \left\{ \begin{array}{ll}
   0 & \text{if\ } j\neq k , \\
   \epsilon \alpha & \text{if\ } j = k = a , \\
   1 & \text{otherwise} ,
  \end{array} \right.
  \label{E:F_def}
 \end{equation}
 where $\alpha > 0$ is a constant to be determined later.  Now I can bound the
 location of the corresponding perturbed eigenvalue of $U V^\epsilon$ by
 applying Gerschgorin circle theorem to the isospectral matrix $F U V^\epsilon
 F^{-1}$.  Specifically, the center $C_j$ and radius $R_j$ of the $j$th
 Gerschgorin circle on the complex plane for the matrix $F U V^\epsilon F^{-1}$
 are given by
 \begin{equation}
  C_j = \left( F U V^\epsilon F^{-1} \right)_{jj} = \left( U V^\epsilon
  \right)_{jj}
  \label{E:G_center}
 \end{equation}
 and
 \begin{eqnarray}
  R_j & = & \sum_{\ell\neq j} \left| \left( F U V^\epsilon F^{-1}
   \right)_{j\ell} \right| \nonumber \\
  & = & \left\{ \begin{array}{ll}
   \epsilon \alpha \sum_{\ell\neq a} \left| \left( U V^\epsilon \right)_{a\ell}
    \right| & \text{if\ } j = a , \\
   \\
   \frac{\left| \left( U V^\epsilon \right)_{ja} \right|}{\epsilon \alpha} +
   \sum_{\ell\neq a,j} \left| \left( U V^\epsilon \right)_{j\ell} \right| &
    \text{if\ } j \neq a .
  \end{array} \right.
  \label{E:G_radius}
 \end{eqnarray}
 Combined with
 \begin{equation}
  \sum_{\ell=1}^n \unorderbarebra{U}{j}\unorderket{V}{\ell}
  \,\unorderbarebra{V}{\ell}\unorderket{U}{k} = \delta_{jk}
  \label{E:completeness}
 \end{equation}
 where $\delta_{jk}$ is the Kronecker delta and
 \begin{equation}
  \left| e^{i x} - 1 \right| = 2 \sin \frac{|x|}{2} \leq |x|
  \text{\ whenever\ } \frac{-\pi}{2} \leq x \leq \frac{\pi}{2} ,
  \label{E:sin_inequality}
 \end{equation}
 I conclude that the distance between the centers of the $a$th and $j$th
 Gerschgorin circles with $j\neq a$ on the complex plane obeys
 \begin{widetext}
 \begin{eqnarray}
  \left| C_a - C_j \right| & \geq & \left| e^{i \unorderangle{U}{a}} - e^{i
   \unorderangle{U}{j}} \right| - \left| \left( U V^\epsilon \right)_{aa} -
   e^{i \unorderangle{U}{a}} \right| - \left| \left( U V^\epsilon \right)_{jj}
   - e^{i \unorderangle{U}{j}} \right| \nonumber \\
  & = & \left| e^{i \unorderangle{U}{a}} - e^{i \unorderangle{U}{j}} \right| -
   \left| e^{i \unorderangle{U}{a}} \sum_{\ell=1}^n \left( e^{i \epsilon
   \unorderangle{V}{\ell}} - 1 \right) \left|
   \unorderbarebra{U}{a}\unorderket{V}{\ell} \right|^2 \right| -
   \left| e^{i \unorderangle{U}{j}} \sum_{\ell=1}^n \left( e^{i \epsilon
   \unorderangle{V}{\ell}} - 1 \right) \left|
   \unorderbarebra{U}{j}\unorderket{V}{\ell} \right|^2 \right| \nonumber \\
  & \geq & \left| e^{i \unorderangle{U}{a}} - e^{i \unorderangle{U}{j}} \right|
   - 2 \sum_{\ell=1}^n \left| e^{i \epsilon \unorderangle{V}{\ell}} - 1 \right|
   \nonumber \\
  & \geq & \left| e^{i \unorderangle{U}{a}} - e^{i \unorderangle{U}{j}} \right|
   - 2 \epsilon \sum_{\ell=1}^n \left| \unorderangle{V}{\ell} \right| .
  \label{E:circle_distance_bound}
 \end{eqnarray}
 \end{widetext}
 And by the same argument, for any $1\leq x\neq y \leq n$,
 \begin{eqnarray}
  \left| \left( U V^\epsilon \right)_{xy} \right| & = & \left| e^{i
   \unorderangle{U}{x}} \sum_{\ell=1}^n e^{i \epsilon \unorderangle{V}{\ell}}
   \unorderbarebra{U}{x}\unorderket{V}{\ell}
   \,\unorderbarebra{V}{\ell}\unorderket{U}{y} \right| \nonumber \\
  & = & \left| \sum_{\ell=1}^n \left( e^{i \epsilon \unorderangle{V}{\ell}} - 1
   \right) \unorderbarebra{U}{x}\unorderket{V}{\ell}
   \,\unorderbarebra{V}{\ell}\unorderket{U}{y} \right| \nonumber \\
  & \leq & \sum_{\ell=1}^n \left| e^{i \epsilon \unorderangle{V}{\ell}} - 1
   \right| \nonumber \\
  & \leq & \epsilon \sum_{\ell=1}^n \left| \unorderangle{V}{\ell} \right| .
  \label{E:off_diagonal_element_bound}
 \end{eqnarray}
 Hence, the sum of radii of the $a$th and $j$th Gerschgorin circles satisfies
 \begin{eqnarray}
  R_a + R_j & \leq & \left[ \frac{1}{\alpha} + \epsilon (n-2) + \epsilon^2
   \alpha (n-1) \right] \sum_{\ell=1}^n \left| \unorderangle{V}{\ell} \right|
   \nonumber \\
  & < & \left[ \frac{1}{\alpha} + \epsilon (n-2) + \epsilon^2 \alpha n \right]
   \sum_{\ell=1}^n \left| \unorderangle{V}{\ell} \right| .
  \label{E:circle_radius_bound}
 \end{eqnarray}
 Since $e^{i \unorderangle{U}{a}}$ is a non-degenerate eigenvalue of the matrix
 $U$, I can always find $\alpha, \epsilon > 0$ satisfying
 \begin{equation}
  \epsilon \alpha \leq \frac{\sqrt{n^2 + 4n} - n}{2n} \label{E:epsilon_value}
 \end{equation}
 and
 \begin{equation}
  \alpha \geq \max_{j\neq a} \frac{2 \sum_{\ell=1}^n \left|
  \unorderangle{V}{\ell} \right|}{\left| \exp \left( i \unorderangle{U}{a}
  \right) - \exp \left( i \unorderangle{U}{j} \right) \right|} .
  \label{E:alpha_value}
 \end{equation}
 With these choices of $\alpha$ and $\epsilon$, the $a$th Gerschgorin circle
 will be disjointed from all the other Gerschgorin circles because $|C_a - C_j|
 > R_a + R_j$ for all $j\neq a$.  So, according to Gerschgorin circle theorem,
 there is exactly one eigenvalue of $U V^\epsilon$ located inside the circle
 centered at $C_a$ and radius $R_a \leq \epsilon^2 \alpha \sum_\ell \left|
 \unorderangle{V}{\ell} \right| = \mbox{O}(\epsilon^2)$ on the complex plane.
 Hence, if eigenvalues of $U$ are non-degenerate, the eigenvalues of $U
 V^\epsilon$ for sufficiently small $\epsilon$ are given by
 \begin{equation}
  \sum_{k=1}^n e^{i \left( \unorderangle{U}{j} + \epsilon \unorderangle{V}{k}
  \right)} \left| \unorderbarebra{U}{j}\unorderket{V}{k} \right|^2 +
  \mbox{O}(\epsilon^2)
  \label{E:eigenvalue_epsilon}
 \end{equation}
 for $j=1,2,\ldots ,n$.

 Eq.~(\ref{E:eigenvalue_epsilon}) also holds even if eigenvalues of $U$ are
 degenerate.  To show this, all I need to do is to modify the above proof as
 follows.  Suppose $e^{i \unorderangle{U}{a}}$ is an $r$-fold degenerate
 eigenvalue of $U$.  And denote ${\mathcal H}$ the $r$-dimensional Hilbert
 subspace making up of the corresponding degenerate eigenvectors of $U$.  For
 $\epsilon$ sufficiently close to $0$, $V^\epsilon$ is strictly diagonally
 dominant.  Hence, the $r\times r$ submatrix $\left. V^\epsilon
 \right|_{\mathcal H}$ formed by retaining only the rows and columns of
 $V^\epsilon$ corresponding to the span of ${\mathcal H}$ is diagonalizable.
 In other words, there exists a basis $\{ \unorderket{U}{j} \}_{j=1}^n$ such
 that $\left. U \right|_{\mathcal H}$ and $\left. V^\epsilon
 \right|_{\mathcal H}$ are simultaneously diagonalized.  Regarding $U$ and
 $V^\epsilon$ as matrices in this basis, I know that
 \begin{equation}
  \unorderbra{U}{j} V^\epsilon \unorderket{U}{a} = \unorderbra{U}{a} V^\epsilon
  \unorderket{U}{j} = 0 \text{\ if\ } j\neq a \text{\ and\ }
  \unorderangle{U}{j} = \unorderangle{U}{a}
  \label{E:degenerate_demand}
 \end{equation}
 provided that $\epsilon > 0$ is sufficiently small.  I now express $U
 V^\epsilon$ in this basis and replace the diagonal matrix $F$ in
 Eq.~(\ref{E:F_def}) by
 \begin{equation}
  F_{jk} = \left\{ \begin{array}{ll}
   0 & \text{if\ } j\neq k , \\
   \epsilon \alpha & \text{if\ } j = k \text{\ and\ } \unorderangle{U}{j} =
   \unorderangle{U}{a} , \\
   1 & \text{otherwise} .
  \end{array} \right.
  \label{E:degenerate_F_def}
 \end{equation}
 Then, if $\unorderangle{U}{j} \neq \unorderangle{U}{a}$, the inequalities in
 Eqs.~(\ref{E:circle_distance_bound}) and~(\ref{E:circle_radius_bound})
 constraining the centers and radii of the $a$th and $j$th Gerschgorin circles
 of the matrix $F U V^\epsilon F^{-1}$ still apply.  Whereas in the case of
 $\unorderangle{U}{j} = \unorderangle{U}{a}$,
 \begin{widetext}
 \begin{eqnarray}
  \left| C_a - C_j \right|^2 & = & \left| \sum_{\ell=1}^n \left( \left|
   \unorderbarebra{U}{a}\unorderket{V}{\ell} \right|^2 - \left|
   \unorderbarebra{U}{j}\unorderket{V}{\ell} \right|^2 \right) e^{i \epsilon
   \unorderangle{V}{\ell}} \right|^2 \nonumber \\
  & = & \left[ \sum_{\ell=1}^n \left( \left|
   \unorderbarebra{U}{a}\unorderket{V}{\ell} \right|^2 - \left|
   \unorderbarebra{U}{j}\unorderket{V}{\ell} \right|^2 \right) \cos \left(
   \epsilon \unorderangle{V}{\ell} \right) \right]^2 + \left[ \sum_{\ell=1}^n
   \left( \left| \unorderbarebra{U}{a}\unorderket{V}{\ell} \right|^2 - \left|
   \unorderbarebra{U}{j}\unorderket{V}{\ell} \right|^2 \right) \sin \left(
   \epsilon \unorderangle{V}{\ell} \right) \right]^2 \nonumber \\
  & = & -2 \sum_{\ell\neq\ell'} \left( \left|
   \unorderbarebra{U}{a}\unorderket{V}{\ell} \right|^2 - \left|
   \unorderbarebra{U}{j}\unorderket{V}{\ell} \right|^2 \right) \left( \left|
   \unorderbarebra{U}{a}\unorderket{V}{\ell'} \right|^2 - \left|
   \unorderbarebra{U}{j}\unorderket{V}{\ell'} \right|^2 \right) \sin^2 \left[
   \frac{\epsilon \left( \unorderangle{V}{\ell} - \unorderangle{V}{\ell'}
   \right)}{2} \right] .
  \label{E:degenerate_circle_distance_bound1}
 \end{eqnarray}
 By Taylor's formula with remainder,
 \begin{equation}
  |C_a - C_j| = B_j \epsilon + \mbox{O}(\epsilon^3)
  \label{E:degenerate_circle_distance_bound2}
 \end{equation}
 where
 \begin{equation}
  B_j = \left[ \frac{-1}{2} \sum_{\ell\neq\ell'} \left( \left|
  \unorderbarebra{U}{a}\unorderket{V}{\ell} \right|^2 - \left|
  \unorderbarebra{U}{j}\unorderket{V}{\ell} \right|^2 \right) \left( \left|
  \unorderbarebra{U}{a}\unorderket{V}{\ell'} \right|^2 - \left|
  \unorderbarebra{U}{j}\unorderket{V}{\ell'} \right|^2 \right) \left(
  \unorderangle{V}{\ell} - \unorderangle{V}{\ell'} \right)^2 \right]^{1/2} .
  \label{E:B_j_def}
 \end{equation}
 \end{widetext}
 It is important to note that $B_j > 0$ if $C_a \neq C_j$ and that the
 magnitude of the $\mbox{O} (\epsilon^3)$ remainder term in
 Eq.~(\ref{E:degenerate_circle_distance_bound2}) is less than or equal to that
 of the $B_j \epsilon$ term provided that $\epsilon$ is sufficiently close to
 $0$.  From Eqs.~(\ref{E:degenerate_demand}) and~(\ref{E:degenerate_F_def}),
 the radius of the $j$th Gerschgorin circle of the matrix $F U V^\epsilon
 F^{-1}$ obeys
 \begin{equation}
  R_j \leq \epsilon^2 \alpha (n - r) \sum_{\ell=1}^n \left|
  \unorderangle{V}{\ell} \right| < \epsilon^2 \alpha n \sum_{\ell=1}^n \left|
  \unorderangle{V}{\ell} \right| \text{\ if\ } \unorderangle{U}{j} =
  \unorderangle{U}{a} .
  \label{E:degenerate_circle_radius_bound}
 \end{equation}
 Suppose the set $\{ j \colon C_a = C_j \}$ has $r' \leq r$ elements.  Then,
 by choosing $\alpha, \epsilon > 0$ satisfying Eq.~(\ref{E:epsilon_value}),
 \begin{equation}
  \alpha \geq \max_{j\colon \unorderangle{U}{j} \neq \unorderangle{U}{a}}
  \frac{2 \sum_{\ell=1}^n \left| \unorderangle{V}{\ell} \right|}{\left| \exp
  \left( i \unorderangle{U}{a} \right) - \exp \left( i \unorderangle{U}{j}
  \right) \right|}
  \label{E:degenerate_alpha_value}
 \end{equation}
 and
 \begin{equation}
  4\epsilon \alpha n \sum_{\ell=1}^n \left| \unorderangle{V}{\ell} \right|
  \leq \min_{j\colon \unorderangle{U}{j} = \unorderangle{U}{a} \text{\ and\ }
  B_j \neq 0} B_j ,
  \label{E:degenerate_additional_value}
 \end{equation}
 the $r'$ Gerschgorin circles with a common center $C_a$ and each with
 $\mbox{O}(\epsilon^2)$ radius will be disjointed from the rest of the
 Gerschgorin circles.  Hence, by Gerschgorin circle theorem, there are exactly
 $r'$ eigenvalues of $U V^\epsilon$ each obeying
 Eq.~(\ref{E:eigenvalue_epsilon}).  Hence, Eq.~(\ref{E:eigenvalue_epsilon})
 also holds for the degenerate eigenvalue case.  I also remark that
 Eq.~(\ref{E:eigenvalue_epsilon}) resembles the Rayleigh-Schr\"{o}dinger series
 truncated at the $\epsilon^2$ terms for time-independent perturbation of
 Hermitian operators.

 The argument of Eq.~(\ref{E:eigenvalue_epsilon}) equals $\unorderangle{U}{j} +
 \arg \left[ \sum_k e^{i \epsilon \unorderangle{V}{k}} \left|
 \unorderbarebra{U}{j}\unorderket{V}{k} \right|^2 + \mbox{O}(\epsilon^2)
 \right]$.  In order words, the arguments of the eigenvalues of $U V^\epsilon$
 obey
 \begin{equation}
  \unorderangle{U V^\epsilon}{j} = \unorderangle{U}{j} + \epsilon \sum_{k=1}^n
  \unorderangle{V}{k} \left| \unorderbarebra{U}{j}\unorderket{V}{k} \right|^2 +
  \mbox{O}(\epsilon^2) \bmod 2\pi .
  \label{E:angle_relation}
 \end{equation}
 Since $\epsilon$ is a small positive number and all arguments are written in
 their principle values, Eq.~(\ref{E:angle_relation}) implies
 \begin{eqnarray}
  \left| \unorderangle{U V^\epsilon}{j} \right| & \leq & \left|
   \unorderangle{U}{j} + \epsilon \sum_{k=1}^n \unorderangle{V}{k} \left|
   \unorderbarebra{U}{j}\unorderket{V}{k} \right|^2 \right| +
   \mbox{O}(\epsilon^2) \nonumber \\
  & \leq & \left| \unorderangle{U}{j} \right| + \epsilon \sum_{k=1}^n \left|
   \unorderangle{V}{k} \right| \,\left| \unorderbarebra{U}{j}\unorderket{V}{k}
   \right|^2 + \mbox{O}(\epsilon^2) .
  \label{E:angle_inequality}
 \end{eqnarray}
 Note that for sufficiently small $\epsilon > 0$, the equality in the first
 line holds if and only if $\left| \unorderangle{U}{j} + \epsilon \sum_k
 \unorderangle{V}{k} \left| \unorderbarebra{U}{j}\unorderket{V}{k} \right|^2
 \right| \leq \pi$.  Now, applying the eigenvalue perturbation and stability
 results for the sum of two diagonalizable matrices in Ref.~\cite{na_book},
 I conclude that the eigenvalues of the positive semi-definite Hermitian matrix
 \begin{eqnarray}
  & & \tilde{H}(U,V,\epsilon) \nonumber \\
  & \equiv & \sum_{j=1}^n \left( \left| \unorderangle{U}{j} \right|
   \,\unorderket{U}{j}\unorderbra{U}{j} + \epsilon \left| \unorderangle{V}{j}
   \right| \,\unorderket{V}{j}\unorderbra{V}{j} \right) \nonumber \\
  & = & \sum_{j=1}^n \left[ \desabsangle{U}{j} \,\desket{U}{j}\desbra{U}{j} +
   \epsilon \desabsangle{V}{j} \,\desket{V}{j}\desbra{V}{j} \right] \nonumber
   \\
  & \equiv & \tilde{H}_a(U) + \epsilon \tilde{H}_b(V)
  \label{E:H_epsilon}
 \end{eqnarray}
 are precisely those given in the last line of Eq.~(\ref{E:angle_inequality}).
 (Similar to the above analysis for the degenerate eigenvalue case, this result
 is also true when eigenvalues of $\tilde{H}_a(U)$ are degenerate.  The key of
 the proof is to carefully pick a basis so that $\tilde{H}_a(U)$ and
 $\tilde{H}_b(V)$ are simultaneously diagonalized for every degenerate
 subspace of $\tilde{H}_a(U)$.)  Combined with Eq.~(\ref{E:angle_inequality}),
 I conclude that
 \begin{equation}
  \desabsangle{U V^\epsilon}{j} \leq \deseig{\tilde{H}(U,V,\epsilon)}{j} +
  \mbox{O} (\epsilon^2)
  \label{E:eigenvalue_dominance}
 \end{equation}
 for $j = 1,2,\cdots , n$.

 According to Corollary~6.6 in Ref.~\cite{eigenvalue_book},
 \begin{eqnarray}
  \sum_{j=1}^m \desabsangle{U V^\epsilon}{j} & \leq & \sum_{j=1}^m
   \deseig{\tilde{H}(U,V,\epsilon)}{j} + \mbox{O}(\epsilon^2) \nonumber \\
  & \leq & \sum_{j=1}^m \left[ \deseig{\tilde{H}_a(U)}{j} + \deseig{\epsilon
   \tilde{H}_b(V)}{j} \right] + \mbox{O}(\epsilon^2) \nonumber \\
  & = & \sum_{j=1}^m \left[ \desabsangle{U}{j} + \epsilon \desabsangle{V}{j}
   \right] + \mbox{O}(\epsilon^2) .
  \label{E:hermitian_analog}
 \end{eqnarray}
 Interestingly, this inequality is the Hermitian analogy of what I need to
 prove here.

 Now, by iteratively applying Eq.~(\ref{E:hermitian_analog}) to $U V = U
 \overbrace{V^{1/q} V^{1/q} \cdots V^{1/q}}^{q \textrm{ terms}}$, I get
 \begin{equation}
  \sum_{j=1}^m \desabsangle{U V}{j} \leq \sum_{j=1}^m \left[ \desabsangle{U}{j}
  + \desabsangle{V}{j} \right] + \mbox{O}(\frac{1}{q}) .
  \label{E:final_inequality}
 \end{equation}
 By taking the limit $q\rightarrow +\infty$, I obtain the second inequality in
 Eq.~(\ref{E:UV_inequality}).
\end{proof}

\begin{remark}
 Note that for each $j = 1,2,\ldots, n$, the equality of
 Eq.~(\ref{E:angle_inequality}) holds if and only if $\unorderangle{U
 V^\epsilon}{j}$, $\unorderangle{U}{j}$ and $\unorderangle{V}{k}$ are all
 non-negative or non-positive for all $k$ whenever
 $\unorderbarebra{U}{j}\unorderket{V}{k} \neq 0$.  And from the proof of
 Corollary~6.6 in Ref.~\cite{eigenvalue_book}, the equality of
 Eq.~(\ref{E:hermitian_analog}) holds if and only if the spans of $\{
 \deseigket{\tilde{H}(U,V,\epsilon)}{j} \}_{j=1}^m$, $\{ \deseigket{\tilde{H}_a
 (U)}{j} \}_{j=1}^m$ and $\{ \deseigket{\tilde{H}_b(V)}{j} \}_{j=1}^m$ agree.
 Suppose the vector $\vec{\mu}$ consists of $r$ distinct $\mu_j$'s.  Then using
 the above two observations and by induction on $r$, one can prove the
 following necessary and sufficient conditions for the second inequality in
 Eq.~(\ref{E:UV_inequality}) of Theorem~\ref{Thrm:angle_bound} to become an
 equality.  The detailed proof is left to interested readers.
 \begin{enumerate}
  \item The $n$-dimensional Hilbert space ${\mathcal H}$ on which $U, V$ and
   $U V$ act can be written as the direct sum $\bigoplus_{j=1}^r
   {\mathcal H}_j$.  Moreover $U$, $V$ and $U V$ are simultaneously block
   diagonalized with respected to this direct sum decomposition of
   ${\mathcal H}$.  That is to say, $\langle\phi|U|\psi\rangle = 0$ whenever
   $|\phi\rangle$ and $|\psi\rangle$ belong to different ${\mathcal H}_j$'s.
   And similarly for $V$ and $U V$.
  \item The ordering of absolute values of the arguments of eigenvalue of $U,
   V$ and $U V$ respects the direct sum decomposition of ${\mathcal H}$ in the
   sense that $\desabsangle{\left. U\right|_{{\mathcal H}_j}}{k} \geq
   \desabsangle{\left. U\right|_{{\mathcal H}_{j'}}}{k'}$ for all $k, k'$
   whenever $j > j'$.  And similarly for $V$ and $U V$.  Furthermore, when
   calculating $\enorm{\cdot}$ using Eq.~(\ref{E:norm_def}), the absolute
   values of the arguments of the eigenvalue in each of the corresponding
   diagonal blocks in $U, V$ and $U V$ are associated with the same value of
   $\mu_j$.
  \item If the $\mu_j$ corresponding to ${\mathcal H}_j$ is non-zero, then
   $\left. U\right|_{{\mathcal H}_j}$, $\left. V\right|_{{\mathcal H}_j}$ and
   $\left. U V\right|_{{\mathcal H}_j}$ can be further simultaneously block
   diagonalized with respected to the direct sum decomposition ${\mathcal H}_j
   = {\mathcal H}_j^+ \oplus {\mathcal H}_j^-$.  Furthermore, $\desangle{\left.
   U\right|_{{\mathcal H}_j^+}}{k}$, $\desangle{\left. V
   \right|_{{\mathcal H}_j^+}}{k}$ and $\desangle{\left. U
   V\right|_{{\mathcal H}_j^+}}{k} \geq 0$ for all $k$; while $\desangle{\left.
   U\right|_{{\mathcal H}_j^-}}{k}$, $\desangle{\left. V
   \right|_{{\mathcal H}_j^-}}{k}$ and $\desangle{\left. U V
   \right|_{{\mathcal H}_j^-}}{k} \leq 0$ for all $k$.
 \end{enumerate}

 Since the number of independent conditions for equality in
 Eq.~(\ref{E:UV_inequality}) increases with the dimension $n$ of the unitary
 matrices as well as the value $m$ used in $\spcenorm{\cdot}{m}$, I conclude
 that as $n$ or $m$ increases, it is harder and harder for the equality to hold
 provided that the unitary matrices are drawn randomly from the Haar measure of
 $U(n)$.  I verify this assertion in Fig.~\ref{F:enorm}, which shows the plots
 of $\spcenorm{U V}{m}$ against $\spcenorm{U}{m} + \spcenorm{V}{m}$.
 \label{Rem:equality_conditions}
\end{remark}

\begin{acknowledgments}
 I thank F.~K. Chow, Y.~T. Lam, K.~Y. Lee, H.-K. Lo and in particular
 C.~H.~F. Fung for their enlightening discussions.  This work is supported by
 the RGC grant number HKU~700709P of the HKSAR Government.
\end{acknowledgments}

\bibliography{qc51.2}
\end{document}